\documentclass{article}
\usepackage{arxiv}

\RequirePackage{tgtermes}
\RequirePackage{newtxtext}
\RequirePackage{newtxmath}
\RequirePackage{bm}
\RequirePackage{endnotes}

\usepackage[utf8]{inputenc} % allow utf-8 input
\usepackage[T1]{fontenc}    % use 8-bit T1 fonts
\usepackage{hyperref}       % hyperlinks
\usepackage{url}            % simple URL typesetting
\usepackage{booktabs}       % professional-quality tables
\usepackage{amsfonts}       % blackboard math symbols
\usepackage{nicefrac}       % compact symbols for 1/2, etc.
\usepackage{microtype}      % microtypography
\usepackage{lipsum}		% Can be removed after putting your text content
\usepackage{graphicx}
\usepackage{natbib}
\usepackage{doi}
\usepackage{amsthm}
\theoremstyle{definition}
\newtheorem{assumption}{Assumption}
\newtheorem{definition}{Definition}
\newtheorem{theorem}{Theorem}
\newtheorem{lemma}{Lemma}
\theoremstyle{plain}
\newtheorem{proposition}{Proposition}

\usepackage{algorithm}
\usepackage{algpseudocode}
\usepackage{natbib}
\bibpunct[, ]{(}{)}{,}{a}{}{,}%
\def\bibfont{\small}%
\newcommand{\minimize}[1]{\displaystyle\minim_{#1}}
\newcommand{\minim}{\mathop{\hbox{\rm minimize}}}
\newcommand{\maximize}[1]{\displaystyle\maxim_{#1}}
\newcommand{\maxim}{\mathop{\hbox{\rm maximize}}}
\newcommand{\sbjt}{\mathrm{subject\ to}}

\usepackage{float}
\usepackage{amssymb}
\usepackage{tabularx,booktabs,caption}
\usepackage{amsmath,amssymb,amsfonts}
\usepackage{graphicx}
\usepackage{textcomp}
\usepackage{xcolor,soul}
\usepackage[english]{babel}
\usepackage{verbatim}
\usepackage{comment}
\usepackage{algpseudocode}
\usepackage{algorithm}

\DeclareMathOperator{\I}{\mathbb{I}}

\newcommand{\mat}[1]{\boldsymbol{#1}}
\renewcommand{\vec}[1]{\boldsymbol{\mathrm{#1}}}
\newcommand{\vecalt}[1]{\boldsymbol{#1}}

\newcommand{\MINone}[3]{\begin{array}{ll} \displaystyle \minimize_{#1} & {#2} \\ \subjectto & {#3} \end{array}}

\newcommand{\prob}{\mathbb{P}}

\newcommand{\bmat}[1]{\begin{bmatrix} #1 \end{bmatrix}}

\providecommand{\mP}{\ensuremath{\mat{P}}}

\providecommand{\vg}{\ensuremath{\vec{g}}}

\makeatother
\usepackage{ulem}
\usepackage{pstricks}

\newcommand{\overbar}[1]{\mkern1.5mu\overline{\mkern-3.0mu#1\mkern-1mu}\mkern 1.5mu}

\providecommand{\vg}{\ensuremath{\vec{g}}}
\DeclareFontFamily{OT1}{pzc}{}
\DeclareFontShape{OT1}{pzc}{m}{it}{ <-> s*[1.15] pzcmi7t }{}
\allowdisplaybreaks

\title{Decentralized Integration of Grid Edge Resources into Wholesale Electricity Markets via Mean-field Games}

\date{} 					% Or removing it

\author{ {\includegraphics[scale=0.06]{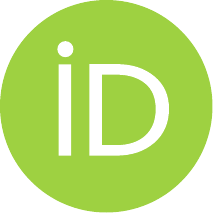}\hspace{1mm}Chen ~Feng}\\
	Edwardson School of Industrial Engineering\\
	Purdue University\\
	West Lafayette, IN 47907 \\
	\texttt{fc123good@gmail.com} \\
	\And
    {\includegraphics[scale=0.06]{orcid.pdf}\hspace{1mm}Andrew L.~Liu} \\
	Edwardson School of Industrial Engineering\\
	Purdue University\\
	West Lafayette, IN 47907 \\
	\texttt{andrewliu@purdue.edu} \\
}

\renewcommand{\shorttitle}{DER Integration via Mean-field Games}

\begin{document}
\maketitle
\begin{abstract}
	Grid edge resources refer to distributed energy resources (DERs) located on the consumer side of the electrical grid, controlled by consumers rather than utility companies. Integrating DERs with real-time electricity pricing can better align distributed supply with system demand, improving grid efficiency and reliability. 
However, DER owners, known as prosumers, often lack the expertise and resources to directly participate in wholesale energy markets, limiting their ability to fully realize the economic potential of their assets. Meanwhile, as DER adoption grows, the number of prosumers participating in the energy system is expected to increase significantly, creating additional challenges in coordination and market participation. 

To address these challenges, we propose a mean-field game framework that enables prosumers to autonomously learn optimal decision policies based on dynamic market prices and their variable solar generation. Our framework is designed to accommodate heterogeneous agents and demonstrates the existence of a mean-field equilibrium (MFE) in a wholesale energy market with many prosumers. Additionally, we introduce an algorithm that automates prosumers' resource control, facilitating real-time decision-making for energy storage management. Numerical experiments suggest that our approach converges towards an MFE and effectively reduces peak loads and price volatility, especially during periods of external demand or supply shocks. This study highlights the potential of a fully decentralized approach to integrating DERs into wholesale markets while improving market efficiency.
\end{abstract}

%\FUNDING{This research was supported by DOE DE-OE0000921 and NSF ECCS-2129631.}

%Supplemental Material:
%Data Ethics & Reproducibility Note:

%\KEYWORDS{Stochastic programming, Decision support,Uncertainty, Disaster response, Optimization}

%\HISTORY{Received: Month DD, YYYY; Accepted: Month DD, YYYY; Published Online: Month DD, YYYY}

\section{Introduction}

The growing adoption of distributed energy resources (DERs), such as rooftop solar panels and energy storage, presents significant opportunities to enhance grid efficiency and resilience. To fully leverage these resources, integrating them into wholesale energy markets is essential for enabling more flexible and reliable grid operations. However, traditional market structures impose high entry barriers for small-scale DER participation due to minimum size requirements and complex market rules. To address these challenges, the Federal Energy Regulatory Commission (FERC) issued Order 2222 \citep{FERC2222}, mandating that DERs be granted access to wholesale markets. While this regulatory change facilitates DER integration, effective mechanisms for small prosumers' (aka DER owners') participation remain unclear. A key concern is that prosumers often hesitate to relinquish direct control of their assets to aggregators -- entities that pool multiple small-scale DERs to meet market size thresholds. Existing literature primarily focuses on how aggregators bid into wholesale markets on behalf of their customers and how contracts for direct load control can be structured, largely overlooking decentralized alternatives that preserve prosumer autonomy.

Our work addresses this gap with three key contributions. First, we develop a fully decentralized framework that enables prosumers to optimize DER operations independently, guided by real-time locational marginal prices (LMPs). Second, we formulate this problem as a mean-field game, where prosumers operate under a mean-field assumption -- treating their individual bids as having negligible impact on LMPs. The collective bids of all prosumers, however, can influence market outcomes. Within this framework, we prove the existence of a mean-field equilibrium (MFE) for an infinite population of agents and an \(\epsilon\)-Markov-Nash equilibrium for a large but finite population. Third, we propose a scalable, low-overhead learning algorithm that allows prosumers to adapt their strategies based on LMP fluctuations, supporting real-time storage management and bid optimization without requiring centralized coordination. Numerical results demonstrate that our approach reduces price volatility and peak loads, even in the presence of supply or demand shocks, thereby improving market stability.

We want to emphasize that our framework is prescriptive rather than descriptive -- it is designed to prescribe optimal prosumer actions rather than merely describe observed behaviors. While our primary focus is electricity markets, the proposed mean-field approach is broadly applicable to any multiagent system where decisions are influenced by a shared external variable, such as market prices or other aggregate signals. Furthermore, the algorithm developed in this work is generalizable and extends beyond energy markets to settings where large-scale agent interactions shape market dynamics in different sectors.

The remainder of this paper is organized as follows. Section \ref{sec:lit} reviews relevant literature on models and algorithms for multiagent bidding in wholesale markets.  Section \ref{sec:Wholesale} outlines the wholesale electricity market model. Section \ref{sec:Prosumer} formulates the prosumer optimization problem. Section \ref{sec:MFG} integrates wholesale market dynamics and prosumer decision-making into a mean-field game framework and establishes the existence of a mean-field equilibrium. Section \ref{sec:algo} describes the proposed learning algorithm. Section \ref{sec:num} presents the numerical experiments and results. Finally, Section \ref{sec:concl} concludes with a summary and potential future research directions.

\section{Literature Review}
\label{sec:lit}
There is extensive literature on individual agents' strategies for bidding into wholesale markets, using either optimization-based or learning-based approaches. In contrast, we focus on systems involving multiple agents. Existing research in this area can be broadly divided into two categories: agent-based simulations and game-theoretic approaches. 
%while individually they are price-takers, their collective actions do affect market prices. There are previous works on the design of wholesale energy market with the integration of DERs and interaction between participants. 
Agent-based simulation (ABS) is widely used to model bidding behaviors in wholesale energy markets, offering a natural approach for studying multi-agent systems. Reviews such as \cite{ABS_Review1, ABS_Review2, ABS_BookChapter} highlight its role in this field, with early works \cite{ABS_Price} and later studies \cite{ArgonneABS1, ArgonneABS2, shafie2014stochastic} advancing the method. A key aspect of ABS is defining appropriate behavioral models for each agent type, creating a heterogeneous artificial economy \citep{ABS_BookChapter}. While agents could be modeled as utility-maximizers considering other agents' actions -- akin to game theory -- ABS often avoids this complexity due to computational challenges.

An alternative, introduced by \cite{Roth1}, uses a simpler adaptive strategy based on action propensities, termed reinforcement learning (RL). Unlike modern RL (as presented in \cite{RLBook}), this model updates action probabilities based on past rewards without state-based feedback or value functions. Despite its simplicity, it effectively predicts human behavior in certain games \citep{Roth2}, inspiring adaptive multi-agent learning studies in energy markets, such as \cite{Bunn, Tesfatsion}. Similar adaptive methods \cite{Ilic, Rogers} explore bidding behavior and dynamic pricing responses but lack the ability to handle intertemporal decisions, such as energy storage management, which is central to this work.

With advancements in modern RL theories and algorithms, multi-agent reinforcement learning (MARL) offers sophisticated methods that avoid preset behavioral assumptions, relying only on utility maximization over time. Naturally, MARL has been applied to model agent participation in energy markets \cite{XueMARL, GoranMARL}. However, MARL still faces two significant challenges: a lack of theoretical guarantees -- specifically, whether multi-agent interactions will converge to an equilibrium, a steady state, or result in chaotic behavior -- in complex environments, and scalability issues, particularly in large systems involving thousands or more agents.

On the game theory side, there is a rich body of literature analyzing bidding strategies and market interactions. Nash-Cournot models, where agents act as quantity setters to maximize their own profits while accounting for market-clearing prices based on total quantities, are widely used in electricity market studies to analyze market power and strategic interactions among generators (for example, \cite{hobbs1986network,willems2002modeling,neuhoff2005network,metzler2003nash}). Another widely used framework is the supply function equilibrium (SFE), where agents compete by submitting supply functions instead of fixed quantities. SFE models are particularly suitable for wholesale electricity markets, as they capture the price-quantity relationship under market clearing (see, for example, \cite{baldick2004theory,rudkevich2005supply,anderson2002using,anderson2005supply,holmberg2010supply}). While these models provide valuable insights, they are inherently static and fail to capture the intertemporal dynamics that are critical in energy systems with energy storage.

Dynamic game-theoretic models address some limitations of static frameworks by incorporating intertemporal decision-making, allowing for the analysis of strategic behaviors over time. These models have been extensively studied in the economics and game-theory literature and applied to examine the strategic behaviors of electricity market participants. For example, works such as \cite{Liu_Collusion2, Liu_Collusion3, fabra2005price, anderson2011implicit} investigate repeated interactions among power producers, focusing on equilibrium concepts of subgame perfect equilibrium. However, these models typically assume complete and perfect information, meaning that all participants have full knowledge of each other's payoff functions, strategies, and the entire history of the game.
In this work, however, the games involve incomplete information, as consumers and prosumers may lack precise knowledge of others' payoff functions, be unable to observe their actions, or not have access to the full history of the game. The standard equilibrium concept for such dynamic games is the Perfect Bayesian Nash Equilibrium (PBNE) (see \cite{GameTheory}). PBNE requires players to update their beliefs using Bayes' rule and to select strategies that maximize their expected payoffs across all possible game histories. While theoretically appealing, PBNE is impractical for real-world applications, as it assumes that agents possess an unrealistic level of strategic sophistication. Moreover, the computational complexity of these models grows significantly as the number of agents increases, making them difficult to scale to large systems.

Mean-field game (MFG) theory\footnote{We focus here on discrete-time MFGs.} offers a promising solution to the challenges posed by dynamic games with incomplete information by approximating interactions among a large number of agents through an aggregate mean-field effect. 
%In the MFG framework, each agent optimizes their expected discounted payoff by assuming a fixed distribution of the average population state, reflecting the behavior of an infinite population.
MFGs provide several advantages over PBNE, particularly in terms of computational tractability and scalability. Extensive research has explored the existence and uniqueness of mean-field equilibria (MFE) (\cite{adlakha2013mean,light2022mean, MarkovNash}). 
%For instance, \cite{adlakha2013mean} demonstrates MFE existence in games with compact state spaces and strategic complementarities, while \cite{adlakha2015equilibria} extends this to countable and unbounded state spaces. Further generalizations to broader state spaces can be found in \cite{light2022mean, MarkovNash}. 
In addition to theoretical advancements, provably convergent algorithms for computing MFE have been developed, including \cite{guo2019learning, gu2024mean, Yang_LearnPlayMFG}.

Building on the theoretical foundations, MFGs have been applied across various domains, including energy markets, where decentralized decision-making and large populations of interacting agents play a critical role. Notable applications include electricity demand management \citep{MFG_DM} and electric vehicle (EV) charging coordination \citep{MFG_EV1, MFG_EV2}. MFGs have also been used to study electricity price dynamics, a topic closely aligned with the focus of this paper. However, all these energy-related MFG applications adopt continuous-time models, which can be impractical for electricity markets where pricing and market clearing occur at discrete intervals (unlike financial markets). Additionally, current models often overlook interactions between energy system operators and the influence of transmission constraints, which are essential factors in determining electricity prices.  we applied discrete-time MFGs to analyze how DERs' decentralized actions influence wholesale electricity markets in \cite{HeLiu24}, though without theoretical foundations.  In this paper, we expand on this by rigorously studying how the collective and decentralized actions of prosumers with solar PVs and energy storage affect wholesale electricity prices. We establish formal conditions for the existence of MFE and propose a scalable heuristic algorithm, making it well-suited for large-scale energy systems integrating increasing levels of decentralized resources. A major advantage of our algorithm is its minimal computational and memory requirements for each agent, unlike distributed optimization methods such as the alternating direction method of multipliers (ADMM), which require solving (proximal) optimization problems at each step. This makes our approach highly scalable and well-suited for large systems with thousands of DERs. By facilitating control automation with low overhead, our algorithm helps unlock tangible benefits for DER owners, supports DER integration into wholesale markets, and enhances scalability. 

Our work addresses critical gaps in the study of decentralized decision-making in energy. A defining feature of our model is the use of continuous state and action spaces (such as energy storage charging/discharging), providing a more realistic representation of prosumer behavior compared to discretized models. We further extend the framework to include multiple heterogeneous agent types, capturing the diversity in prosumer characteristics and decision-making. This aspect draws inspiration from \cite{mondal2022approximation}, which incorporates heterogeneity in a mean-field control (MFC) setting. However, while their work focuses on cooperative agents, our model explores non-cooperative interactions. This market-driven approach extends beyond energy markets to any domain where aggregate behavior shapes price signals. Leveraging this structure, the highly scalable heuristic algorithm developed in this work can be applied across a wide range of settings.

%Despite its theoretical rigor, game-theoretic modeling often requires strong assumptions about agent rationality and information availability. Bridging the gap between the descriptive flexibility of agent-based simulation and the analytical rigor of game-theoretic models remains an important challenge. This paper addresses this gap by combining elements of learning-based agent models with game-theoretic approaches to tackle intertemporal decision-making in multi-agent energy markets.

\section{Wholesale Energy Market Operations}
\label{sec:Wholesale}
\begin{figure*}[!ht]
	\begin{center}
		\includegraphics[scale=0.11]{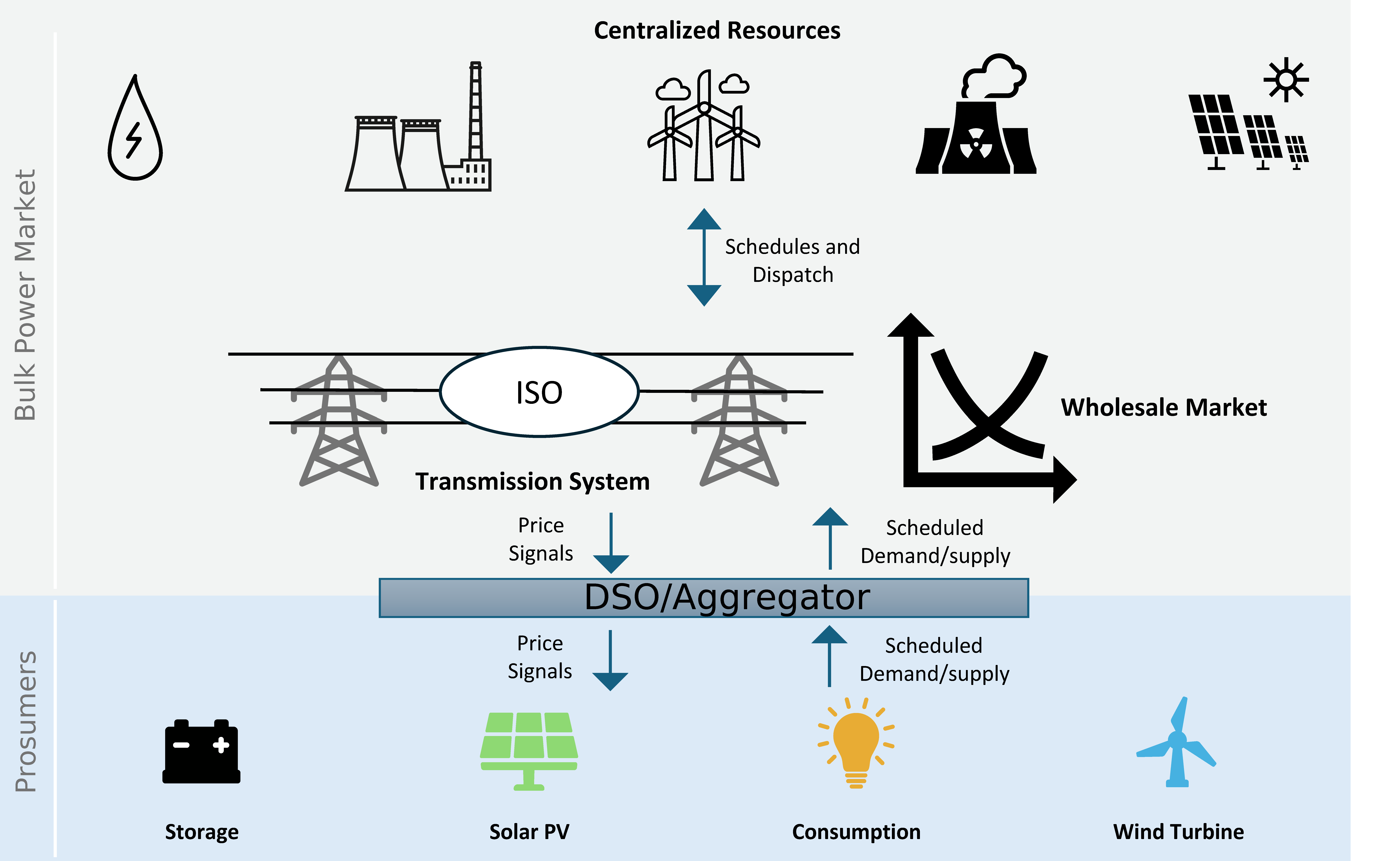}
		\caption{Conceptual framework of a wholesale energy market with aggregators participation}
		\label{fig:MktOper}
	\end{center}
\end{figure*}
In this section, we first present the optimization problem solved by an Independent System Operator (ISO) for a wholesale energy market. As illustrated in Fig. \ref{fig:MktOper}, in each time period (such as an hour), the ISO collects supply and demand bids and runs optimization problems to match the supply and demand with the lowest cost, subject to various engineering and transmission network constraints. The marginal costs of supplying one additional unit of demand at each node, known as the locational marginal prices or LMPs, can be calculated based on the shadow prices of constraints on supply-demand balancing and transmission capacity constraints. We show a key result in this section regarding the Liptschitz continuity of the LMPs with respect to energy demand.   
%Both for larger prosumers and for smaller aggregators, FERC order 2222 requires that all RTOs/ISOs allow for Distributed Energy Resources (DERs) that can provide a maximum of at least 0.1 MW to participate in all markets.  

%\subsection{Hour-ahead real-time Wholesale Market Model and Economic Dispatch Problem}
%\label{sec:EDP}

Consider a (bidirectional) power system network with $N$ nodes and $L$ transmission lines. For simplicity, we assume that each node $n \in \{1,...,N\}$ has only one energy supplier\footnote{If there are multiple suppliers, we can assume that each supplier is on a separate node with such nodes connected by transmission lines of unlimited capacities.} with an operation cost function $C_n(\cdot)$. The node has $I_n$ agents, including both consumers and prosumers. 
%Let $H (=24)$ denote the total number of hours in a day. 
One hour before the operating period \( t \), the \( i \)-th agent at node \( n \) submits its energy demand/supply bid, \( b^{n}_{i,t} \), for period \( t \) to the ISO.
 Note that $b^{n}_{i,t}$ represents the net demand. For consumers, this value is simply their actual energy demand with $b^{n}_{i,t} > 0$. For prosumers, \( b^{n}_{i,t} \) represents net energy demand, calculated as actual demand plus the energy used for charging batteries, minus PV generation and any energy withdrawn from batteries. This net demand can be either positive or negative, where a negative value indicates that the prosumer is supplying energy back to the grid. The decision-making problem for prosumers' bidding is presented in the next section.

Throughout the paper, we make the blanket assumption that the total net demand of all agents in the entire system is positive, that is, 
$\sum_{n=1}^N\sum_{i=1}^{I_n} b^{n}_{i,t} > 0,\quad \forall \ t \in \{1,...,\}.
$
This assumption is reasonable, as it will likely take considerable time in the future before prosumers can meet all consumer energy demands and still produce surplus energy.

In the hour-ahead market, the ISO solves an optimization problem, known as economic dispatch (ED), for the upcoming operating period $t$. This optimization determines the amount of real power to be dispatched from each electric power-generating resource to match the total demand as follows:
\begin{align}
\minimize{\vg_t} &\ \sum_{n=1}^{N} C_n(g^n_{t}) \label{obj1}\\
\sbjt &\ \sum_{n=1}^{N} g^n_{t} \geq \sum_{n=1}^{N}\sum_{i=1}^{I_n} b^{n}_{i,t}  \label{const1}\\ 
& \ -\widehat{\mathrm{F}}_l \leq \sum_{n=1}^{N} \mathrm{PTDF}_{l, n} (g^n_{t}- \sum_{i=1}^{I_n} b^{n}_{i,t})\leq \widehat{\mathrm{F}}_l, \ \forall l \in \{1,...,L\} \label{const2}\\
&\ \ 0 \leq g^n_{t} \leq \widehat{\mathrm{G}}_n, \quad \forall n \in \{1,...,N\}, \label{const3}
\end{align}
where $\vg_t:= \{g^n_{t}\}_{n=1}^N$ is the collection of decision variables, representing the energy generation at node $n$ in time period $t$. 
%and $C_n(\cdot)$ in the objective function represents the corresponding generator's cost function. 
Constraint \eqref{const1} specifies that the total supply must not be less than the total demand, often referred to as the supply and demand balancing constraint. 
%In cost minimization, this constraint typically holds with equality at the optimal solution, unless other constraints prevent this. 
Constraint \eqref{const2} represents the transmission network capacity constraints, with the capacity limit denoted by $\widehat{\mathrm{F}}_l$.  The network, which is assumed to be a connected graph, is modeled as a hub-spoke network, in which energy sent from node $n$ to $n'$ is assumed to be routed from $n$ to a hub (an arbitrary node in the system) first and from the hub to $n'$. The parameter $\mathrm{PTDF}_{l,n}$ in \eqref{const2} represents the power transfer distribution factor, which indicates the fraction of power injected at node $n$ that flows through line $l$.\footnote{For simplicity, we ignore transmission losses in this formulation. However, they can be incorporated as long as the resulting formulation remains a convex optimization problem. In that case, all results presented in this work still hold.
}  Last, $\widehat{\mathrm{G}}_n$ in \eqref{const3} represents the generation capacity for the power plant at node $n$. 

To write out the exact formula of nodal electricity prices, aka the LMPs, we first  use $\mathcal{L}$ to denote the Lagrangian function of the ED problem. For the ease of argument, we use $B^t_h$ to denote the aggregate demand at node $n$ in time period $h$; that is, $B^n_t =  \sum_{i=1}^{I_n}b^n_{i,t}$. 
Let $\lambda$ denote the dual variable associated with constraint \eqref{const1}, and $\overline{\mu}_l$ and $\underline{\mu}_l$ be the dual variables corresponding to  \eqref{const2}. Then the LMPs, denoted by $P^n_t(B^1_{t},\dots, B^{N}_{t})$ for $n = 1, \ldots, {N}$ at time $t$,  are the derivatives of the Lagrangian function with respect to the demand: 
\begin{align}
\label{LMP} 
\mathrm{LMP}^n_{t} & :=  P^n(B^1_{t},\dots,B^{N}_{t}) = \frac{\partial \mathcal{L}}{\partial B^n_{t}} 
= \lambda - \sum_{l=1}^L \mathrm{PTDF}_{l, n} (\overline{\mu}_l - \underline{\mu}_l).
\end{align} 
To establish the main result of this paper, which is the existence of an MFE of the multiagent system, it is crucial to prove that the LMPs are Lipschitz continuous with respect to the demand vector $\mathbf{B}_t := (B^1_t, \dots, B^{N}_t)$. Achieving this requires an assumption regarding the constraint qualification for the ED problem. We state this assumption below and then present the Lipschitz continuity result.
\begin{assumption}
	\label{assump:LICQ} 
	(LICQ) Let \(X(\mathbf{B}_t)\) denote the feasible region of the ED problem \eqref{obj1} -- \eqref{const3}. Define the set \(\mathcal{F}_B\) such that for all \(\mathbf{B}_t \in \mathcal{F}_B\), \(X(\mathbf{B}_t) \neq \emptyset\). We assume that for all \(t\) and for all \(\mathbf{B}_t \in \mathcal{F}_B\), the linear independence constraint qualification (LICQ) holds at all points in \(X(\mathbf{B}_t)\).
	
\end{assumption}
\begin{proposition} \label{prop:LMP_LipCont}
	Assume that the generation cost function $C_n(\cdot)$ in \eqref{obj1} is a strongly convex quadratic function in the form of $C_n(g) = \frac{1}{2}\alpha_n g^2 + \beta_n g + \gamma_n$, with $\alpha_n > 0$  for all $n = 1, \ldots, N$. 
	Under Assumption \ref{assump:LICQ}, with $\mathbf{B}_t\in\mathcal{F}_B$, the LMP at each node $n =1, \ldots, N$, $P^n(\mathbf{B}_t)$, is a single-valued function and Lipschitz continuous with respect to $\mathbf{B}_t$.
\end{proposition}
The proof is in Online Appendix \ref{Proof_LipCont}.

\section{A Prosumer's Markov Decision Process}
\label{sec:Prosumer}
The previous section focuses on the system operator's optimization problem. In this section, we shift the focus to how individual agents participate in a wholesale market. We first introduce a model for a single agent who makes {\it repeated} decisions regarding the charging and discharging of their energy storage over time, in response to real-time pricing tied to the LMPs.  The agents make their decisions under the assumption that the system is in an MFE due to the large number of agents. We then show that an MFE can indeed emerge with heterogeneous agents holding this belief. This is a direct extension of our earlier work in \cite{ZhaoLiu_MAB} where each agent solves a multiarmed bandit problem, which cannot accommodate intertemporal decisions. 

\subsection{Assumptions on the Agents}
\label{subsec:Assump}
To accommodate agents' heterogeneity, we assume that each consumer or prosumer has a type $\theta \in \Theta$, with $\Theta$ being a finite set. These types can include characteristics such as location (e.g., agents at different nodes in the transmission network belong to different types), varying solar PV capacities, battery capacities or types, and distinct load profiles. We assume that agents of the same type are homogeneous in their payoff function, state transition function, battery capacity, PV generation profile, and load distribution. Specifically, 
each agent of type $\theta$ has a battery capacity $\overline{e}^{\theta}=\frac{\overline{C}^{\theta}}{I^{\theta}}$, where $I^{\theta}$ is the number of agents of type $\theta$, and $\overline{C}^{\theta}$ is the aggregated battery capacity of all type-$\theta$ agents. This definition ensures that as $I^{\theta}$ approaches infinity, each individual’s capacity becomes infinitesimally small, yet the aggregate capacity remains well-defined and finite.

%Assume that the number of agents in each type is large enough so that a single agent's effect on the outcome of the system is negligible. 
\subsection{Single-agent's Dynamic Optimization}
\label{subsec:SDP}
In the following, we provide the key elements in building an individual agent's decision-making model, with a given agent type $\theta \in \Theta$.

\textit{Action.} At each time period \( t \), agent \( i \) determines the fraction of energy to charge or discharge from their battery, expressed as a percentage of battery capacity and denoted by \( a_{i,t} \in \mathcal{A} := [-1, 1] \).
 A positive value of $a_{i,t}$ signifies a charging action, whereas a negative value indicates discharging. 
%(on top of the charging/discharging loss) at the beginning of period $h$ on day $t$. 1 means charging energy that equals the battery capacity and -1 means discharging energy that equals the battery capacity.

\textit{State.} The state of an agent consists of three elements: the net load, the state of charge (SoC) of the energy storage, and time of day. The net load, which is a random variable, is defined as the firm (or inflexible) demand minus the energy output from solar PVs. We assume that agents of the same type share an identical daily net load shape, representing the expected value of the net load at the corresponding time of the day. Let \( Q^{\theta}_{i,t} \) denote the net load for agent \( i \) at time period \( t \), where \( Q \) is used to represent `quantity.'
 Since actions (and later, the SoC) are defined as percentages, it is convenient to consider net load as a percentage as well. We introduce the ratio \(q^{\theta}_{i,t} := Q^{\theta}_{i,t}/{\bar{e}^{\theta}}\), where \(\bar{e}^{\theta}\) is the storage capacity as defined in Section \ref{subsec:Assump}.  The transition from $q^{\theta}_{i,t}$ to $t+1$ is assumed to be purely driven by weather conditions and by random noise, which accounts for variations in real-time electricity usage among agents. Mathematically speaking, we have that 
\begin{equation}\label{eq:q_decompose}
q^{\theta}_{i,t} = \omega^{\theta}_t + \zeta^{\theta}_{i,t}, 
\end{equation} where both \(\omega^{\theta}_t\) and \(\zeta^{\theta}_{i,t}\) are random variables. The first term, \(\omega^{\theta}_t\), represents weather-related randomness and is location-specific (depending on the type \(\theta\)) but not agent-specific (hence, no agent index \(i\)).
The second term, \(\zeta^{\theta}_{i,t}\), represents agent-specific random noise in electricity demand. Both variables are assumed to have compact supports, as each agent's electricity demand and PV/storage capacity are finite. 

For the SoC of energy storage, we use $e_{i,t} \in \mathcal{E} \equiv [0, 1]$ to denote the fraction of remaining energy in the battery at the beginning of period $t$ for agent $i$. The state transition of the SoC, denoted by $E(\cdot,\cdot)$, can be expressed as:
\begin{align}
e_{i,t} & :=  E(e_{i,t-1},a_{i,t-1})  = 
\max \{\min \{e_{i,t-1} + a_{i,t-1},1\}, 0\}, \ t = 1, 2, \ldots, \label{eq:StateTransition}
\end{align}
with the `max' and `min' operations to ensure that the actions will not lead to over-charging or over-discharging of the battery. In the following, we use $\mathcal{E} = [-1,1]$ to denote the general space of SoC. 
%\hl{When $h = 1$, representing the first period of day $d$, the state transition occurs from the last period of the preceding day, as detailed in the first part of the state-transition equation} \eqref{eq:StateTransition}. \hl{To streamline our discussion and avoid repeatedly addressing the cases of $h = 1$ and $h = 2, ..., H$ separately in the rest of the paper, we slightly abuse the notation. Specifically, for the subindex ${i, d, h}$, when $h = 1$, it should be understood that it refers to the index of ${i, d-1, H}$.} 
%In the following discussions, we use $\mathcal{E}$ to denote the generic feasible space in terms of SoC. Since it represents the percentage of energy in storage, it is the same across all agents. Therefore, we can omit the agent subindex. 

For the time of day, $h = 1, \ldots, H$, we account for the fact that an agent's policy should vary throughout the day. For example, even if the state of charge of the energy storage is the same at 10 AM and 6 PM in a day, the corresponding optimal strategy should be different. At 10 AM, solar energy is generally abundant, and household electricity usage is typically lower, as many people are at work. In this case, a good strategy might be to charge the battery to full. On the other hand, at 6 PM, solar energy is diminishing, and people are returning home, causing residential energy demand to increase. Even with the decrease in commercial and industrial loads at that time, the overall energy demand is expected to rise in the early evening. Hence, a good strategy at this time might be to discharge the energy storage. For a general time period index $t$, we use $T_{day}(t)$ to denote the time of day of $t$, and denote the set of times of day as $\mathcal{H}$, where $\mathcal{H} = \{1, \ldots, H\}$.

In the following, we treat \( e_{i,t} \) and \( T_{day}(t) \) as the state variable, denoted by the generic notation \( s_{i,t} \), while considering the random variable \( q_{i,t}^{\theta} \) as exogenous. Note that the state of charge transition equation \eqref{eq:StateTransition} does not involve any uncertainties; that is, the net load \( q_{i,t}^{\theta} \) does not directly affect the state transition. This is a specific modeling choice, which we will explain further after introducing the agents' bid functions in \eqref{eq:bids}. This deterministic transition simplifies both the analysis and algorithm design in the subsequent sections. Additionally, the transition of the time of day is trivially deterministic. We use \( Tr(s, a) \) to denote the general state transition, which includes both the state of charge transition and the time of day transition, which simply increments by one (that is, moves to the next time of day). It is understood that when \( T_{day}(t) = H \), the time of day cycles back to 1 in $t+1$, representing the start of the next day.

%\hl{Don't treat $Q_{i,t}$ as state; treat it as an exogenous uncertainty. Also, may need to add the time of the day as a state variable. (In this way, the algorithm consider situations such as in a heat wave, at 10 AM, solar energy may be ample, but around 7 PM, demand will peak as people get to home from work and all turning up ACs, but solar is phasing out. So the agent may choose not to discharge. But if the current time is 6 PM, then optimal $a$ may be to discharge in the next hour)}

\textit{Charging/Discharging Efficiency.} For most energy storage batteries, both charging and discharging efficiencies decrease as the respective rates increase. As demonstrated in Figure \ref{fig:efficiency} (taken from \citep{amoroso2012advantages}), the efficiency of a lithium-ion battery approximately follows a linear relationship with the charging rate.
\begin{figure}[htb]
	\centering
	\includegraphics[scale=0.3]{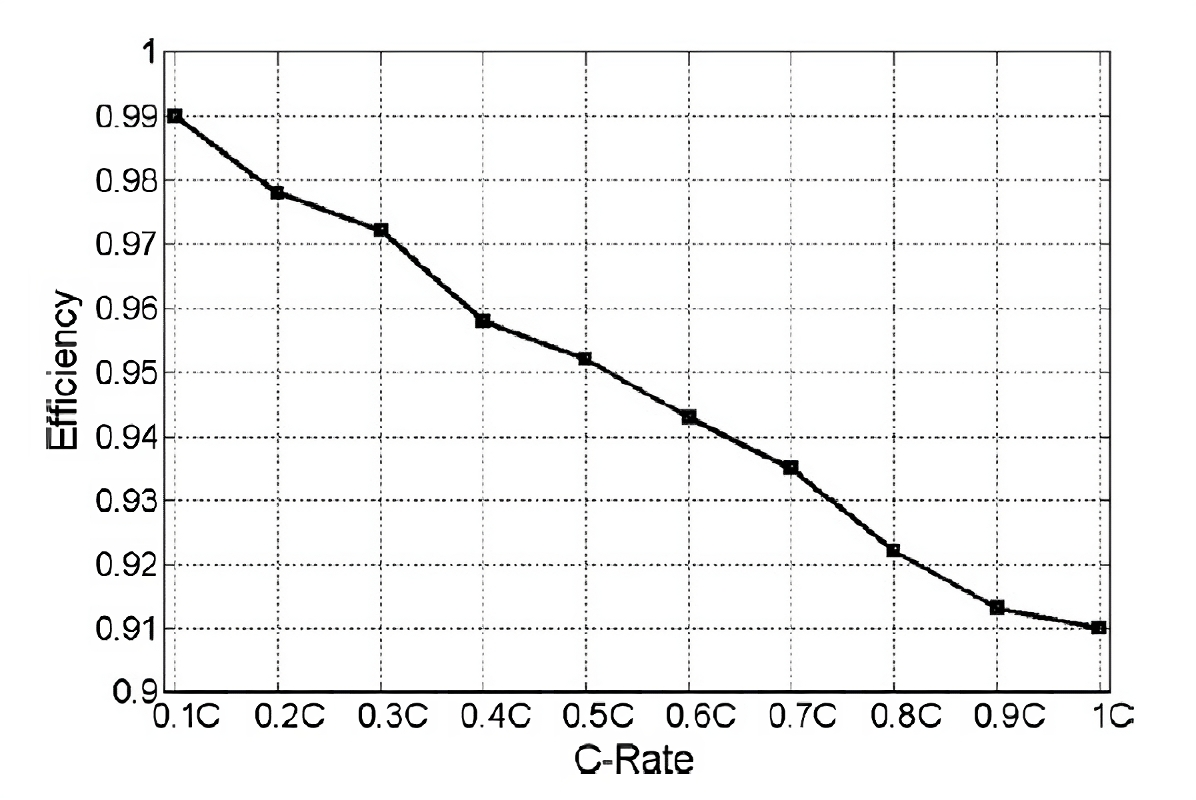}
	\caption{Experimental results showing the dependency of charging efficiency on the charging rate for a Li-ion cell ($C$ represents battery capacity)\cite{amoroso2012advantages}}  \label{fig:efficiency}
\end{figure}
We model the energy charged to or discharged from the battery at a constant rate of $\frac{a_{i,t}}{\Delta t}$ for each period, where $a_{i,t}$ is agent $i$'s charging/discharging action as defined earlier, and $\Delta t$ represents the duration of the period. Consequently, we assume that charging efficiency decreases linearly as \( a_{i,t} \) increases for \( a_{i,t} \in [0,1] \), while discharging efficiency increases linearly as \( a_{i,t} \) decreases for \( a_{i,t} \in [-1,0] \).
 The efficiency function is defined as:
\begin{align}
\eta(a_{i,t}) =\begin{cases}
\alpha_0 + \alpha_d \cdot a_{i,t}, & \text{if} \ a_{i,t} < 0, \\
\alpha_0 - \alpha_c \cdot a_{i,t}, & \text{if} \ a_{i,t} \geq 0, 
\end{cases}
\label{eq:eff}
\end{align}
where $\alpha_0 \in (0,1)$ represents the baseline charging/discharging efficiency. The coefficients $\alpha_c > 0$ and $\alpha_d > 0$ represent the rates at which efficiency decreases with increasing charging and discharging percentages, respectively. To ensure that efficiencies across all $a_{i,t} \in [-1,1]$ are non-negative, we impose the conditions that $\alpha_0 - \alpha_c > 0$ and $\alpha_0 - \alpha_d > 0$.

\textit{Energy bid.} With solar panels and energy storage, a prosumer's bid, \(b^{\theta}_{i,t}\), can be represented as a function of the state variable $s_{i,t}$, action $a_{i,t}$, and the exogenous uncertainty $q^{\theta}_{i,t}$:
\begin{align}\label{eq:bids}
& b^{\theta}_{i,t}(s_{i,t}, a_{i,t}, q^{\theta}_{i,t}) =  \left\{
\begin{aligned}
& q^{\theta}_{i,t}\cdot \overline{e}^{\theta} + \eta(a_{i,t})\overline{e}^{\theta} \cdot \max \big\{-e_{i,t},\ a_{i,t} \big\}, \ \text{if  } a_{i,t} < 0\ \mathrm{(discharging)},\\[8pt]
&  q^{\theta}_{i,t} \cdot \overline{e}^{\theta} +  \frac{\overline{e}^{\theta}\cdot \min \big\{1-e_{i,t},\ a_{i,t}\big\}}{\eta(a_{i,t})}, \ \text{if} \ a_{i,t} \geq 0\ \mathrm{(charging)}. 
\end{aligned}
\right.
\end{align}
The bids represent the sum of the net load and battery charging or discharging quantity (after accounting for efficiency losses). Since the state variables, action variables, and exogenous uncertainties -- $e_{i,t}$, $a_{i,t}$, and $q^{\theta}_{i,t}$ -- are all bounded, the bid $b_{i,t}$ is also bounded for all $i$ and $t$.

The formulation in \eqref{eq:bids} defines the bidding strategy. The first case (\( a_{i,t} < 0 \), discharging) indicates that the agent first uses its energy storage to meet its net energy demand, \( q^{\theta}_{i,t} \cdot \overline{e}^{\theta} \), measured in absolute terms rather than as a percentage. If there is excess energy after discharging, it is sold directly into the wholesale market. Conversely, if there is a shortfall, the agent purchases the required energy from the wholesale market.  
The second case (\( a_{i,t} \geq 0 \), charging) states that the bid represents the total energy purchased from the grid to meet the agent's energy demand plus the energy charged to storage. This bid formulation makes the state transition in \eqref{eq:StateTransition} deterministic and simplifies the analysis considerably. While this is not the only way to design a bidding strategy, it has the advantage of giving prosumers precise control over the energy levels they wish to maintain in their storage.

The downside of this approach is that it assumes any excess supply or demand from prosumers can always be absorbed by or met in the wholesale market. This assumption holds when the collective size of prosumers is small relative to the overall grid’s supply and demand, or when considering the geographical averaging effect -- where excess energy from prosumers in one area can be used to meet the needs of another. However, as the number of prosumers increases, this assumption may become problematic, especially since most prosumers have solar generation, not wind. Unlike wind energy, which benefits from geographical diversity due to varying wind conditions, solar power generally does not. During the day, solar energy is generated across all locations (with some variation due to irradiance), but in the evening, production drops to zero. A more sophisticated bidding strategy using reinforcement learning can be explored as a future research direction.

\textit{Population Profile.} Before detailing the payoff functions for each agent, it is essential to establish the concept of a population profile, which aggregates the states and actions of all agents. In a large population game, although the actions of an individual agent do not directly affect the payoffs of others, the aggregated actions of the entire group do. 
%A complicating factor in this context is the time of day -- for example, the population profile at 10 AM is unlikely to be the same as at 6 PM. To account for this variation, we introduce a mapping from time $t$ to the time of day, $T_{day}(t)$, where $T_{day}(t) \in \{1, 2, \ldots, H\}$, representing, for example, the 24 hours in a day if $H=24$.
Population profiles vary by both time of day and agent type. We begin by defining the empirical distribution of population profiles for a finite number of agents at time  $t$. To account for general state and action spaces, we use $\mathcal{B}(X)$ to denote the Borel $\sigma$-algebra of a generic set $X$. Then, for a state space \( S \in \mathcal{B}(\mathcal{S}) \), defined as the Cartesian product of \( \mathcal{E} \) and \( \mathcal{H} \) -- the SoC space and the set of all times of the day -- and an action space \( A \in \mathcal{B}(\mathcal{A}) \), 
we define the following:
\begin{align}\label{eq:EmpiricalPopProfile}
& p^{I^{\theta}}_{t}(S,A) = \frac{1}{I^{\theta}}\sum_{i=1}^{I^{\theta}} \sum_{h\in\mathcal{H}}\mathbb{I}_{\{e_{i,t}\in \mathcal{E}\}}\times \mathbb{I}_{\{T_{day}(t)=h\}}\times \mathbb{I}_{\{a_{i,t} \in A\}},
\end{align}
where $\mathbb{I}_{\{e_{i,t}\in S\}}$, $\mathbb{I}_{\{T_{day}(t)=h\}}$, and $\mathbb{I}_{\{a_{i,t} \in A\}}$ are indicator functions that respectively track the state of charge, time of day, and action of agent $i$. This formulation represents the empirical joint distribution of states and actions across the population. 

Let $p^{\infty,\theta}_{h}$ be the limit as $I^{\theta}, t \rightarrow \infty$ for all $\theta$ and $h \in \mathcal{H}$. This limit represents a probability distribution over the joint state and action space, denoted by $\Xi := \mathcal{S} \times \mathcal{A}$. We use $\mathcal{P}(\Xi)$ to denote the set of all probability measures on $\Xi$, and let $p^{\infty}_{h} := [p^{\infty,\theta}_{h}]_{\theta \in \Theta} \in \mathcal{P}(\Xi)^{|\Theta|}$ denote the population profile of all types at time $h$ of a day, with $|\Theta|$ being the cardinality of the type space $\Theta$. Furthermore, we use $p^{\infty}$ to denote the collection of $p^{\infty}_{h}$ for all $h \in \{1, \ldots, H\}$; that is, $p^{\infty} := [p^{\infty}_{h}]_{h=1}^{H} \in \mathcal{P}(\Xi)^{|\Theta|\times H}$.

\textit{Payoff.} The single-stage payoff function for a type-\(\theta\) agent at time \(t\), with a long-run equilibrium of the population profile $p^{\infty}_{_{T_{day}(t)}}$, is denoted as \(R_{i,t}^{\theta}(s_{i,t}, a_{i,t}, q^{\theta}_{i,t} \mid p^{\infty}_{_{T_{day}(t)}}): \mathcal{S} \times \mathcal{A} \times \mathcal{Q} \rightarrow \mathbb{R}\).
To provide the explicit mathematical formulation of the payoff function, we first define $P^{n(\theta)}_t(\cdot): \mathcal{P}(\Xi)^{|\Theta|} \to \mathcal{R}$ as a function that maps the population profile at time $t$ to the LMP at node $n$ in the transmission network. With a slight abuse of notation, we use $n(\theta)$ to denote the location within the transmission network where agents of type $\theta$ are situated. 
%(Recall that the same type of agents is assumed to be located at the same node.) 
The stage payoff function is:
\begin{align} 
& R^{\theta}_{i,t}(s_{i,t}, a_{i,t}, q^{\theta}_{i,t}|p^{\infty}_{_{T_{day}(t)}}) =  - P^{n(\theta)}_t(p^{\infty}_{_{T_{day}(t)}}) \times b^{\theta}_{i,t}(e_{i,t}, a_{i,t}, q^{\theta}_{i,t}), \label{eq:StagePayff} 
\end{align}
where $b^{\theta}_{i,t}$ is agent $i$'s energy bid at time $t$, as specified in \eqref{eq:bids}. Since $b_{i,t} < 0$ indicates energy sales to the grid, this formula yields a positive payoff for the agent, while an energy purchase bid ($b_{i,t} > 0$) results in a cost, or a negative payoff, to the agent. 

Note that the stage payoff is a random variable due to the stochastic nature of the LMPs, demand, and variable renewable outputs. When determining optimal policies, agents must rely on the expected value of the payoff. Therefore, to simplify the notation and analysis, we directly define the expected payoff and denote it by \( \overbar{R}^{\theta}_t(s, a \mid p^{\infty}) \). Since each individual agent's bid is small (infinitesimal in the case of an infinite number of agents), we assume that the individual bid does not impact the LMPs and is thus independent of them. Consequently, we can write out the expected value of the payoff as follows:
\begin{align}
& \overbar{R}^{\theta}_t(s, a \mid p^{\infty}) :=  \ \mathbb{E} \left[ R^{\theta}(s, a, q^{\theta} \mid p^{\infty}) \right] \label{eq:EPayoff}  = \ - \ \overbar{P}^{n(\theta)}_t(p^{\infty}_{_{T_{day}(t)}}) \times \mathbb{E}_{q^{\theta}} \left[ b^{\theta}_{i,t}(e_{i,t}, a_{i,t}, q^{\theta}_{i,t}) \right], 
\end{align}
where \( \overbar{P}^{n(\theta)}_t(p^{\infty}_{_{T_{day}(t)}}) \) represents the expected LMP at node $n(\theta)$ and time $t$.\\[-10pt] 

\noindent\textbf{Remark 1.} (Boundedness of $\overbar{R}^{\theta}_t$.)
Note that we assumed both the net load and energy storage capacity of each agent are bounded. Therefore, each agent's bid is bounded, regardless of external uncertainties. For net load, following a similar approach to how we define individual energy storage capacity, we assume that the total net load for each agent type $\theta$ is bounded by an upper limit $\overbar{Q}^{\theta}$. As a result, the aggregate demand at each time $t$, $\mathbf{B}_t = (B^1_t, \ldots, B^N_t)$, lies within a compact region. By the Lipschitz continuity of the LMPs with respect to $\mathbf{B}_t$ (under the assumption that the LICQ holds at all the feasible points), the LMPs are uniformly bounded (over the feasible region of $\mathbf{B}_t$). Hence, the payoff function $R^{\theta}_{i,t}(s_{i,t}, a_{i,t}, q^{\theta}_{i,t} \mid p^{\infty}_{_{T_{day}(t)}})$, along with its expected value, is also uniformly bounded. \\[-10pt]

\noindent\textbf{Remark 2.} (Continuity of $\overbar{R}^{\theta}_t$.) Based on the formulation of an agent's bid in \eqref{eq:bids}, for a given $e \in S$, the function is continuous with respect to the action $a$. This is evident because the bid function consists of two parts: one for $a > 0$ and the other for $a < 0$. In both cases, the max and min functions are continuous, and so is the charging/discharging efficiency function \eqref{eq:eff}. Hence, their product is continuous as well. At $a = 0$, whether approaching from $a \to 0^+$ or $a \to 0^-$, the bid function $b^{\theta}_{i,t}(e_{i,t}, a_{i,t}, q^{\theta}_{i,t})$ always converges to $q_{i,t}^{\theta} \cdot \overbar{e}^{\theta}$. This reflects the fact that as the action approaches zero (i.e., no charging or discharging), the bid approaches the net load $q_{i,t}^{\theta} \cdot \overbar{e}^{\theta}$. Therefore, for each $e \in S$, the expected payoff function $\overbar{R}^{\theta}_t(s, a \mid p^{\infty})$ is continuous with respect to $a$.

%the beliefs held by all agents of type $\theta$ regarding the LMPs at hour $h$. We will show the properties of this belief in the next section and propose a simple method for agents to update the beliefs in Algorithm \ref{Algo_meanfield}.

\begin{comment}
\begin{assumption}
\label{assumpLMP}
$\phi^h^{\theta}(s_h^{\infty})$ is uniformly bounded from both above and below for $\forall s^{\infty} \in \mathcal{P}(\mathcal{S})^{|\Theta|}$.
\end{assumption}
\end{comment}

%Under Assumption \ref{assumpLMP}, $R^h$ is uniformly bounded from both above and below.

\subsection{Dynamic Optimization and Optimal Policy} 
The repeated decision-making problem of how to submit quantity bids and manage energy storage to maximize a prosumer's  long-term payoff can be modeled as a stochastic dynamic programming (SDP) problem. Specifically, each agent $i$ aims to maximize the following expected discounted payoff over an infinite time horizon:
\small
\begin{align}
\sup_{\pi_{i,t}}  \mathbb{E}\left[\sum_{t=0}^{\infty} \beta^t 
R^{\theta}_{i,t}(s_{i,t}, a_{i,t}, q^{\theta}_{i,t} \mid p^{\infty}_{t}) \,\bigg|\, a_{i,t} \sim \pi_{i,t},\ s_{i,0},\  p^{\infty}_{0}\right],   \label{eq:SDP}
\end{align}
\normalsize
where $\beta \in (0,1)$ is the discount factor. The stochastic process begins with an initial individual state \( s_{i,0} \) and population profile \( p^{\infty}_{0} \). At time $t = 0, 1, \ldots$ , the agent selects an action \( a_{i,t} \) according to a policy \( \pi_{i,t} \).

%While in a classic approach, people try to identify the conditions under which a stationary policy exists, and then focus on the stationary  

Assume that the population profile is already at an equilibrium (its existence is the main subject of Section~\ref{sec:MFG}). Since the only actions are energy storage charging and discharging -- which we model as percentages -- the action space is compact, irrespective of the state. As discussed in Remarks~1 and~2, the expected stage reward function is bounded and continuous with respect to the actions. Additionally, the state transition function~\eqref{eq:StateTransition} is continuous with respect to the action~$a$. Therefore, by the well-established result in~\cite{Puterman_MDP} (Theorem~6.2.12), an optimal stationary policy of~\eqref{eq:SDP} exists. Furthermore, according to a well-known result in stochastic dynamic programming~\cite{BertsekasDP} (Proposition~1.2.3), a stationary optimal policy must satisfy the Bellman equation.

To formulate the Bellman equation, we first define the value function for type-$\theta$ agents, which depends on both an agent's individual state and the population profile. For simplicity, we remove the agent index $i$ here, but still keep the type index $\theta$. 
%As stated earlier, time of day is a complicating factor. Within the same time of a day, there can be a stationary optimal policy, but from one hour to another hour, the optimal policy will be different. Hence, each agent has a value vector, as opposed to a single value function. Let $\mathbf{s}^0$ denote vector of first $H$ hours' state variables; that is, $\mathbf{s}^0 = (s_1, \ldots, s_H)^T$. Then the value vector can be defined as 
%$$V^{\theta}(\mathbf{s}^0, \mathbf{p}^{\infty})= \bmat{V^{\pi^{\theta}}_1(s_1, p^{\infty}) \\ \vdots \\ V^{\pi^{\theta}}_H(s_H, p^{\infty})}$$
With  a population profile $p^{\infty}=[p^{\infty}_{h}]_{h=1}^{H}$ and a stationary policy $\pi^{\theta}$, the expected discounted present value for each state variable $s \in S$ can be expressed as follows:
\begin{align}
&  V^{\pi^{\theta}}(s, p^{\infty}) = \mathbb{E}\Bigg[\sum_{t=0}^{\infty} \beta^{t} 
R^{\theta}_{t}(s_{t}, a_{t}, q^{\theta}_t|p^{\infty})\bigg| a_{t}\sim\pi^{\theta} ,\ s_0\Bigg]. 
\end{align}
Let $V^{\pi^{\theta^*}}(s, p^{\infty}) = \max_{\pi^{\theta}\in\Pi^{\theta}}  V^{\pi^{\theta}}(s, p^{\infty})$, where $\Pi^{\theta}$ is the set of all admissible policies of the type-$\theta$ agent. 
%\hl{This should be $\pi^{\theta}$ over all hours. Then this is a vector-valued objective function. How do you know a single set of optimal policies $\bm{\pi^{\theta}} = (\pi^{\theta}_1, \pi^{\theta}_2, \ldots, \pi^{\theta}_{H}, \ldots)$ exists that can simultaenously optimize all the elements in the objective function? May use Weierstrass Theorem, but the set of all admissible policies is not necessarily compact. May check the proof on the existence of optimal policy on a normal DP.}
The well-known Bellman equation can then be written as,
\begin{align}
&V^{\pi^{\theta^*}}(s, p^{\infty}) = \ \max_{a\in \mathcal{A}} \left\{\overbar{R}^{\theta}(s, a|p^{\infty}) + \beta  V^{\pi^{\theta^*}}\left[Tr(s,a),\ p^{\infty}\right]\right\}, \label{eq:Bellman} 
\end{align}
where the function $\overbar{R}^{\theta}(s, a|p^{\infty})$ represents the expected value of a one-period payoff, and  $Tr(s,a)$  is the general state transition for both the storage's SoC and the time of day. Note that, as emphasized earlier, the state transition is deterministic. Therefore, the Bellman equation does not require a transition probability density function to describe how the state evolves.
%\textcolor{blue}{Note that there is no expected value over $V^{\pi^{\theta^*}}(s, p^{\infty})$ in the objective function in the above formula (this is not true: the definition of V^* already has a built-in expected value) because the state transition is deterministic, and so is the population transition then (the policy is also not stochastic)}
The corresponding optimal policy mapping is  
\begin{align}
& \Pi^{\theta^*}(s, p^{\infty}) \equiv   \arg\max_{a\in\mathcal{A}} \left\{\overbar{R}^{\theta}(s, a|p^{\infty}) + \beta  V^{\pi^{\theta^*}}\left[Tr(s,a),\ p^{\infty}\right]\right\}. \label{eq:OptPolicy}
\end{align}

%\subsection{Single-valuedness of an Agent's Optimal Policy}
In the following, we show a key property regarding the agents' optimal stationary policy, which is crucial for proving the existence of an MFE in the next section.

\begin{proposition}
	Under the assumptions of Proposition \ref{prop:LMP_LipCont}, for an agent of type $\theta$, the optimal stationary policy mapping $\Pi^{\theta^*}(s,\ p^{\infty})$ is single-valued and continuous  with respect to $(s,\ p^{\infty})$. 
	\label{prop:unique}
\end{proposition}
The proof is in Online Appendix \ref{subsec:PolicyUnique}.

\section{Multiagent Mean-field Games}
\label{sec:MFG}
In this section, we first provide the precise definition of an MFE and show its existence in the context of DER integration into a wholesale electricity market. We then provide an algorithmic approach that enables agents to adaptively learn how to play in the mean-field game, which facilitates fully decentralized control automation.

\subsection{Mean-Field Equilibrium: Definition and Existence}

The essence of an MFE in this context is that each agent assumes the LMPs are at a long-run equilibrium and believes their individual actions do not influence this equilibrium. Based on this assumption, each agent selects an optimal strategy, which collectively leads to an equilibrium consistent with the assumed LMPs. This state is known as an MFE. A more precise definition of MFE is provided below.

%In a mean-field game with the assumption of an infinite number of agents, it is possible that the population profile becomes a fixed distribution over time when the game reaches a steady state, and will not be influenced by a single agent's behavior. Therefore, an equilibrium for a mean-field game can be described informally as follows: each agent assumes that the population profiles are fixed at some certain distribution in different periods of the day, and makes optimal decisions based on the assumed population profile. Given the optimal strategy of each player, the actual population profile in each period turns out to be consistent with the assumption. In this section, we will first study agents' optimal strategy given certain profiles in our wholesale market model, and then show the result of the existence of MFE of the model.

%Let $\phi^{\theta}(s^{\infty}):= [\phi^1^{\theta}(s^{\infty}_1),\cdots,\phi^H^{\theta}(s^{\infty}_H)]$ be the vector of LMP functions over different periods for type $\theta$ agents, and $G:= \{G^{\theta}\}_{\theta \in \Theta}$ be the collection of optimal strategy functions for all types of agent. Now we can define the MFE in our wholesale market model, 

\begin{definition}
	\label{def:MFE}
	A collection of stationary strategy $\pi^{*} := [\pi^{\theta^{*}}]_{\theta \in \Theta}$ and a population profile $p^{\infty}:=\big[[p^{\infty,\theta}_1]_{\theta \in \Theta},\cdots,[p^{\infty,\theta}_{H}]_{\theta \in \Theta}\big] \in \mathcal{P}(\mathcal{S})^{|\Theta|\times H}$ constitute an MFE if for each $\theta \in \Theta$ and $h = 1, \ldots, H$,  the following two conditions hold: 
	\begin{itemize}
      \setlength{\itemindent}{0.5em}
		\item Optimality: for a given state $s \in \mathcal{S}$, $\pi^{\theta^{*}} \in \Pi^{\theta^*}(s, p^{\infty})$ as defined in \eqref{eq:OptPolicy}.
		
		\item Consistency: for all $S \times A \in \mathcal{B}(\mathcal{S}) \times \mathcal{B}(\mathcal{A}$), where $\mathcal{B}(\cdot)$ is the Borel algebra of the corresponding set, and $s \in S$, 
		\begin{align}
		\begin{split}
		&p^{\infty,\theta}_h(S \times A) = \int_{S \times A} \I_{S \times A}\bigg\{E\left(e,\pi^{\theta^*}_{h-1}\left(e, p^{\infty}\right) \right), \pi_h^{\theta^*}\Big( E\left(e,\pi^{\theta^*}_{h-1}(e,p^{\infty})\right), p^{\infty})\Big)\bigg\} \, dp^{\infty,\theta}_{h-1}(s,a), \label{eq:Invariant}
		\end{split}
		\end{align}
	\end{itemize}
\end{definition} 
\noindent where $E(e, a)$ represents the state transition function for the energy storage's state of charge,  as in \eqref{eq:StateTransition}. 
In \eqref{eq:Invariant}, when \( h = 1 \), it is understood that the model interprets \( h-1 \) as \( H \), which represents the final time period of the previous day. Additionally, with a slight abuse of notation, we use \( \pi^{\theta^*}_{h}(e, p^{\infty}) \) to denote the policy at the state where the time of day is \( h \); that is \(\pi^{\theta^*}_{h}(e, p^{\infty}) := \pi^{\theta^*}\big(s = (e, T_{day} = h),\ p^{\infty}\big) \).

Definition \ref{def:MFE} implies that under an MFE, the population profile at the same time of day on different days remains invariant when each agent adopts an optimal strategy according to \eqref{eq:OptPolicy}. Equivalently, $(\pi^*,p^{\infty})$ is an MFE if and only if $p^{\infty}$ is a fixed point of the MFE operator $\Phi:\mathcal{P}(\mathcal{S})^{H\times|\Theta|} \to \mathcal{P}(\mathcal{S})^{H\times|\Theta|}$ defined by:
\begin{align}
\begin{split}
\label{MFEOPT}
\Phi(p^{\infty})(S \times A)_{\theta \in \Theta}
= &\bmat{[\Phi^{\theta}_1(p^{\infty})(S \times A)]_{\theta \in \Theta}\\
	\vdots\\
	[\Phi^{\theta}_{H}(p^{\infty})(S \times A)]_{\theta \in \Theta}},\\ 
\end{split}
\end{align}
where 

\begin{align}
\begin{split}
&\Phi^{\theta}_h(p^{\infty})(S \times A)= \int_{S \times A} \I_{S \times A}\bigg\{E\left(e,\pi^{\theta^*}_{h-1}(e,p^{\infty}) \right),\pi_h^{\theta^*}\Big( E\left(e,\pi^{\theta^*}_{h-1}(e,p^{\infty})\right), \ p^{\infty})\Big)\bigg\} \, dp^{\infty}_{h-1}(s,a),\\
&\text{for } h=1\cdots H,\text{ and } \forall \theta \in \Theta.
\end{split}
\end{align}
Therefore, to show the existence of an MFE, we will prove that there is a fixed point of $\Phi$, to be presented in Proposition \ref{prop:MFE_exist} below. 

\begin{proposition}[Existence of an MFE]\label{prop:MFE_exist}
	Under the assumptions in Proposition \ref{prop:LMP_LipCont}, an MFE, as defined in Definition \ref{def:MFE}, exists for the prosumer bidding game of direct participation in a wholesale electricity market.
\end{proposition}
The proof uses the Schauder-Tychonoff Fixed Point Theorem; the details are provided in Online Appendix \ref{subsec:Proof_MFEExist}.

\subsection{Finite Agents and Approximate Markov-Nash Equilibrium}

The existence of an MFE established in the previous subsection assumes an infinite number of agents. A natural question arises: What happens when the number of agents is large but finite? More specifically, if each agent in a finite system adopts the mean-field equilibrium policy, which was derived under the infinite-agent assumption, how does this affect the system’s equilibrium properties? 

To address this question, we first formally define the Markov-Nash equilibrium and its approximate counterpart, the $\epsilon$-Markov-Nash equilibrium. For notational convenience, we omit the type index $\theta$ corresponding to an agent $i = 1, \ldots, I$. For a finite number of agents $I$, let $M^i$ denote the set of all Markov policies for agent $i$, and define the Cartesian product $M^\mathcal{I} := \Pi_{i=1}^{I} M^i$. Let $\bm{\pi^{\mathcal{I}}} \in M^\mathcal{I}$ denote the collection of policies of the $I$ agents, i.e., $\bm{\pi^{\mathcal{I}}} = (\pi_1, \ldots, \pi_I)$. The empirical population profile at time $t$, denoted by $p_t^{\mathcal{I}}$, is defined as in \eqref{eq:EmpiricalPopProfile}; that is, $p_t^{\mathcal{I}}(\cdot, \cdot) = [p_t^{\theta}]_{\theta\in\Theta}.$
The initial state of each agent is given by $s_{i,0} = (e_{i,0}, T_{day}(0))$, where the initial state of charge $e_{i,0}$ follows a distribution in $\mathcal{P}[0,1]$, and the initial time of the day $T_{day}(0)$ is arbitrary. The initial states of different agents are assumed to be independent.
Let $\bm{\pi}^{\mathcal{I}}_{-i}$ denote the collection of policies for all agents except agent $i$. The discounted total reward for agent $i$ in a finite-agent game is defined as:
\begin{align}
& J_i^{\mathcal{I}}(\pi^{\mathcal{I}}_i, \bm{\pi}^{\mathcal{I}}_{-i}) \ = \mathbb{E}\left[  \sum_{t=0}^{\infty} \beta^t 
R^{\theta}_{i,t}(s_{i,t}, a_{i,t}, q^{\theta}_{i,t} \mid p^{\mathcal{I}}_{t}) \,\bigg|\, a_{i,t} \sim \pi^{\mathcal{I}}_{i} \hspace*{-2pt}, s_{i,0},  p^{\mathcal{I}}_{0}\right].  \nonumber
\end{align} 

\begin{definition} (Definition 2.2, \cite{MarkovNash}) 
	A policy  $\bm{\pi}^{\mathcal{I}^*} \in M^{\mathcal{I}}$ is a Markov-Nash equilibrium if 
	\begin{equation}\label{eq:MNash}
	J_i^{\mathcal{I}}(\pi^{\mathcal{I^*}}_i, \bm{\pi}^{\mathcal{I}^*}_{-i}) = \sup_{\pi_i \in M_i} J_i^{\mathcal{I}}(\pi_i, \bm{\pi}^{\mathcal{I^*}}_{-i}),\ i=1, \ldots, I. 
	\end{equation}
	It is an $\epsilon$-Markov-Nash equilibrium if 
	\begin{equation}\label{eq:MNash_epsilon}
	J_i^{\mathcal{I}}(\pi^{\mathcal{I^*}}_i, \bm{\pi}^{\mathcal{I}^*}_{-i}) \geq \sup_{\pi_i \in M_i} J_i^{\mathcal{I}}(\pi_i, \bm{\pi}^{\mathcal{I^*}}_{-i}) - \epsilon,\ i=1, \ldots, I. 
	\end{equation}
\end{definition}

It has been established in \cite{MarkovNash} that for any given $\epsilon > 0$, an $\epsilon$-Markov-Nash equilibrium exists when the number of agents $I$ is sufficiently large. However, this result relies on several technical conditions that may be challenging to verify in general settings.
In our case, the specific modeling of the deterministic state transition for the SoC of energy storage significantly simplifies the verification of these conditions. In the following, we demonstrate that these assumptions hold in our framework, justifying the use of the MFE policy even in a finite-agent game. 
%Consequently, the MFE policy leads to an $\epsilon$-Markov-Nash equilibrium when the number of agents is sufficiently large.

\begin{proposition}[$\epsilon$-Markov-Nash Equilibrium]\label{prop:epsilon_MNE}
	Under Assumption \ref{assump:LICQ}, for any \( \epsilon > 0 \), there exists a positive integer \( I(\epsilon) \) such that for all \( I \geq I(\epsilon) \), the policy \( \pi^{(I)} = (\pi_1, \pi_2, \ldots, \pi_I) \), where each \( \pi_i \) is defined as in \eqref{eq:OptPolicy} \ for \( i = 1, \ldots, I \), constitutes an \( \epsilon \)-Markov-Nash equilibrium for the game involving \( I \) prosumers participating in a wholesale energy market.
\end{proposition}
\begin{proof}{Proof}
	To prove the result, we need to verify that the required nine conditions, as outlined in \cite{MarkovNash}, are satisfied under our model. 
	The continuity of the one-stage reward function with respect to the state, action, and population profile, as established in the previous section, along with the compactness of the action space and the boundedness of the reward function (as stated in Remark 1), ensures that several key conditions are satisfied. Furthermore, the specific modeling of the state transition in \eqref{eq:StateTransition}, which is deterministic and independent of the population profile, automatically satisfies the remaining related to the transition dynamics. %\Halmos
\end{proof}

\section{Learning in a Mean-Field Game: An Algorithmic Approach}
\label{sec:algo}
While previous results establish the existence of an MFE, they do not provide a direct strategy for agents to follow in games with a large number of participants. Various RL-based methods have been proposed to approximate the fixed-point iteration needed to converge to an MFE \citep{guo2019learning,guo2023general}. These approaches typically employ a double-loop structure: first, the population profile is fixed while each agent solves an MDP using an RL algorithm, such as Q-learning; then, the population profile is updated.

As highlighted by \cite{Yang_LearnPlayMFG}, the double-loop approach presents two major challenges: (1) the population profile usually evolves simultaneously with agents' policy updates, and (2) for large state spaces, function approximations (such as neural networks) used to represent value and policy functions make solving each subproblem computationally demanding. To address these issues, \cite{Yang_LearnPlayMFG} proposed a single-loop, online algorithm. In their method, the mean-field state is updated via a single step of gradient descent, while agents' policies are updated by one step of policy optimization, informed by real-time feedback from the game.
To ensure convergence, the algorithm employs a fictitious play mechanism, where agents probabilistically update their policies or do nothing. This mitigates instability by smoothing the learning process, allowing the mean-field state to evolve more gradually alongside policy updates.

While single-loop methods improve computational efficiency and stability compared to the double-loop approach, both methods still rely on a fundamental assumption: agents must have access to the global population profile, which is a probability distribution. However, this assumption is often impractical in real-world settings.
In contrast, we propose an approach that leverages a specific feature of our setting -- namely, fluctuations in electricity prices (aka the LMPs) reflect underlying changes in the population profile as well as external uncertainties. Instead of requiring direct access to the population profile, we allow agents to form beliefs about future LMPs at different times of the day. These beliefs guide agents' actions by solving their SDP problems, eliminating the need for explicit knowledge of the population profile. Agents then update their beliefs adaptively based on realized prices, facilitating decentralized and scalable decision-making.

Specifically, let $\widetilde{P}_{h}$ denote the agent's belief about the LMP at time of day \( h \), and \( P_{T_{day}(t) = h}\) is the actual price observed at period \( t \). The belief update rule for the \( h \)-th time of day is given by:
\begin{equation}
\label{eq:AdaptiveRule}
\widetilde{P}_{h} \leftarrow \widetilde{P}_{h} - \delta \cdot (\lfloor t/H\rfloor+1)^{-0.5}(\widetilde{P}_{h} -P_{T_{day}(t) = h}).
\end{equation}
The parameter \( \delta \) is a learning rate in \( (0,1) \), and \( \lfloor t/H \rfloor \) accounts for the total number of days elapsed. This rule adjusts an agent's belief using a diminishing step size, ensuring that recent observations have a greater impact while older data becomes less influential over time. 
%The belief-based approach also aligns well with energy markets, as LMPs are always calculated by the ISO in advance -- such as hour-ahead LMPs in the California market.\footnote{\url{https://www.caiso.com/documents/scheduling-priorities-and-export-schedules-overview.pdf}}

With an agent's belief and a given state \( e_{t,h} \), an optimal decision \( a^* \) is determined by solving:
\begin{align}
a^* =\  & \arg\max_{a \in [-e_{t,h}, 1-e_{t,h}]} 
-\eta(a) \cdot \min\{a,0\} \cdot \widetilde{P}_h  -\max\{a,0\} \cdot \frac{\widetilde{P}_h}{\eta(a)}
+ \beta V_{h+1}(e_{t,h} + a, \widetilde{\mathbf{P}}), \label{eq:OptimalAction}
\end{align}
where \( \eta(a) \) is the charging/discharging efficiency and \( \beta \) is the discount factor, as defined earlier, and $\widetilde{\mathbf{P}}$ denotes the vector of LMP beliefs for all times of the day, given by $(\widetilde{P}_h )_{h=1}^H$. The value function $V_{h+1}$ -- despite being defined for a single period $h+1$ -- depends on LMP beliefs across all time periods.

%The specific method used to calculate $V_{h+1}$ is not constrained by our approach. In the simplest case, it could involve just a single iteration of a value function-based method, similar to the one-step policy update in the single-loop online algorithm proposed in \cite{Yang_LearnPlayMFG}. In our numerical implementation, we discretize the continuous state to approximate value functions and solve \eqref{eq:OptimalAction} using these approximations. 

At this stage, the proposed method remains heuristic; yet numerical experiments consistently demonstrate convergence to a steady state. This suggests that underlying theoretical convergence properties may exist, which we leave as an avenue for future research.

To enhance our algorithm's realism and demonstrate the demand response capability of DERs without direct load control, we incorporate demand and supply shocks to reflect real-world conditions, such as unexpected fluctuations in energy demand or renewable generation. For scenarios involving these shocks, we assume the ISO issues an emergency signal one hour before the event. During these periods, agents adapt their actions based on alternative sets of LMP beliefs corresponding to the type of shock. The value function and optimal decision rules are computed separately for regular and shock periods. Specifically, during a demand shock -- when demand surges and supply is likely insufficient, risking blackouts if no action is taken -- agents replace their regular LMP belief \( \widetilde{P} \) with \( \widetilde{P}^{\text{DS}} \). In contrast, during a supply shock -- where electric power supply likely exceeds demand, such as at night when wind energy surges but electricity demand is low -- agents switch to the LMP beliefs \( \widetilde{P}^{\text{SS}} \). Agents then use these alternative beliefs to determine optimal actions under supply or demand shocks, following the same approach as in \eqref{eq:OptimalAction}. To track the frequency of such events, agents maintain counters for demand and supply shocks, denoted as \( \tau_d \) and \( \tau_s \), respectively. After observing the actual LMP \( P_t \), agents update the LMP beliefs for supply or demand shocks using the same adaptive rule as in \eqref{eq:AdaptiveRule}, with the counters \( \tau_d \) and \( \tau_s \) replacing \( \lfloor t/H \rfloor \) for demand and supply shocks, respectively. 

The pseudocode summarizing the single agent's value-iteration algorithm, including responses to supply and demand shocks, is presented in Algorithm 1. Using the LMP beliefs in \eqref{eq:OptimalAction}, the problem becomes a typical SDP problem. We do not specify a particular algorithm for solving \eqref{eq:OptimalAction}, nor is exact computation required. Therefore, approximate dynamic programming methods, such as those in \cite{BertsekasDP, Powell}, are all applicable. In an extreme case, the algorithm can involve just one step of a gradient-descent-like method to enable a single-loop, online approach. This flexibility makes the framework scalable and adaptable, allowing agents to learn and act effectively in a mean-field game setting, where the mean-field is reflected through market prices. In our implementation, we introduce a small probability for each agent to restart in a random state with random LMP beliefs, a process referred to as regeneration. This serves two purposes. First, it ensures the multi-agent system remains active and adaptive, allowing agents to continue learning. When some agents regenerate, they must relearn, preventing the system from becoming static once it reaches a mean-field equilibrium. Second, it simulates real-world scenarios, where some agents may leave the market while new agents enter, reflecting the natural turnover in such systems.

\begin{algorithm}
\caption{Single Agent's Value-Iteration Algorithm}
\label{Algo_meanfield}
\begin{algorithmic}[1]
\State \textbf{Initialization:}
\State Randomly initialize \( e_{0,1} \in [0,1] \) (initial battery state).
\State Randomly initialize \( \widetilde{P} \) and shock-specific beliefs \( \widetilde{P}^{\text{DS}}, \widetilde{P}^{\text{SS}} \).
\State Set learning rate \( \delta \in (0,1) \); initialize shock counters $\tau_{d} = 0 $ and  $\tau_{s} = 0$.
\vskip8pt
\For{ \( t = 0, 1, \dots \)}
    \For{\textbf{each hour} \( h = 1, \dots, H \)}
            \If{no emergency signal received}
                \State Submit bid \( a^* \) based on \eqref{eq:OptimalAction} using the LMP belief \( \widetilde{P} \). 
                \State Update beliefs using \eqref{eq:AdaptiveRule} after the market price \( P_t \) is observed. 
            \ElsIf{demand shock signal received}
                \State Submit bid \( a^* \) based on \eqref{eq:OptimalAction} using the LMP belief \( \widetilde{P}^{\text{DS}} \).
                 \State Update beliefs after the market price \( P_t \) is observed as follows 
                \begin{equation*}
\widetilde{P}^{DS}_{h} \leftarrow \widetilde{P}^{DS}_{h} - \delta \cdot (\tau_d+1)^{-0.5}(\widetilde{P}^{DS}_{h} - P_{t}).
\end{equation*}
              \State Set $\tau_d \leftarrow \tau_d+1$.
%              \algstore{bkbreak}
%\end{algorithmic}
%\end{algorithm}
%\addtocounter{algorithm}{-1}
%\begin{algorithm}[h]
%\caption{(continue)}
%\begin{algorithmic}[1]
%\algrestore{bkbreak}
            \ElsIf{supply shock signal received}
                \State Submit bid \( a^* \) based on \eqref{eq:OptimalAction} using the LMP belief \( \widetilde{P}^{\text{SS}} \).
                \State Update beliefs after the market price \( P_t \) is observed as follows 
                \begin{equation*}
\widetilde{P}^{SS}_{h} \leftarrow \widetilde{P}^{SS}_{h} - \delta \cdot (\tau_s+1)^{-0.5}(\widetilde{P}^{SS}_{h} - P_{t}). 
\end{equation*}
              \State Set $\tau_s \leftarrow \tau_s+1$.
            \EndIf
    \EndFor
\EndFor
\end{algorithmic}
\end{algorithm}

\section{Numerical Results}
\label{sec:num}
In this section, we apply Algorithm 1 to a test power system comprising both bulk generators and thousands of prosumers and consumers. Our objectives are as follows: (1) to assess if the algorithm can achieve convergence numerically; (2) to observe whether, upon convergence, the algorithm encourages the desired behavior of charging during peak sunshine hours and discharging in the evening when demand increases, ultimately smoothing LMPs and reducing volatility; and (3) to evaluate the algorithm's performance under random supply and demand shocks.
\subsection{Test Case}
In our experiment, we use the IEEE 14-bus system\footnote{Power Systems Test Case Archive -- 14 Bus Power Flow Test Case (\url{https://labs.ece.uw.edu/pstca/pf14/pg_tca14bus.htm}).} as the test case. We assume there is one generator at each bus, with each generator’s total generation cost represented by a quadratic function: \( C_n(g) = \frac{1}{2}\alpha_n g^2 + \beta_n g \). The parameters \( \alpha_n \) and \( \beta_n \) are chosen uniformly from the ranges \([0.0118, 0.0684]\$/\text{MW}^2\text{h}\) and \([150, 233]\$/\text{MWh}\), respectively, based on data from \cite{krishnamurthy20158}, for all $n$. Each power plant is assumed to have a $600$ MW capacity. 
%The coefficient for a generator with fuel oil from \cite{krishnamurthy20158} is applied to each generator. 
All transmission lines' capacity is set to be 1,000 MW. Each node (bus) contains two types of agents: prosumers with DERs (solar, small wind and energy storage) and pure consumers. 

\begin{figure}[!htb]
    \centering
    \includegraphics[scale =0.5]{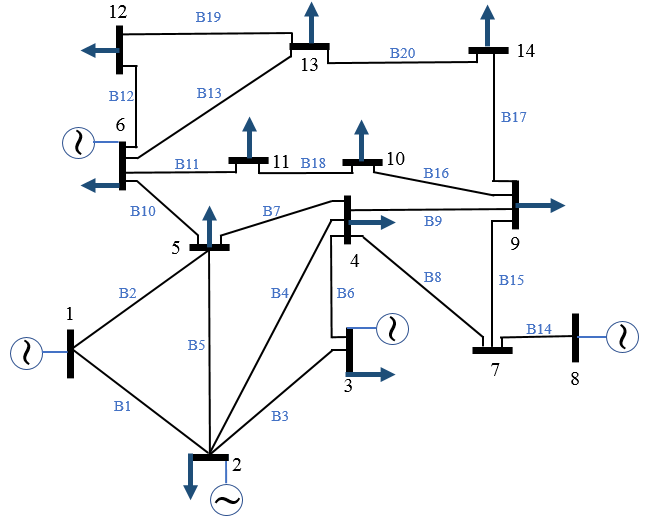}
    \caption{IEEE-14 test bus system}
    \label{fig:ieee14}
\end{figure}
For demand, we use CAISO data,\footnote{Both aggregate load and net load data are available at \url{https://www.caiso.com/todays-outlook##section-net-demand-trend}.} which includes both total aggregated load and net load, where net load is defined as gross load minus distributed wind and solar generation. Using CAISO's 2022 data, we compute two aggregated load shapes -- one for gross load and one for net load -- each representing the average 24-hour profile, as shown in Fig. \ref{fig:shape}.  
\begin{figure}[!htb]
%\FIGURE
\centering
    \includegraphics[scale=0.7]{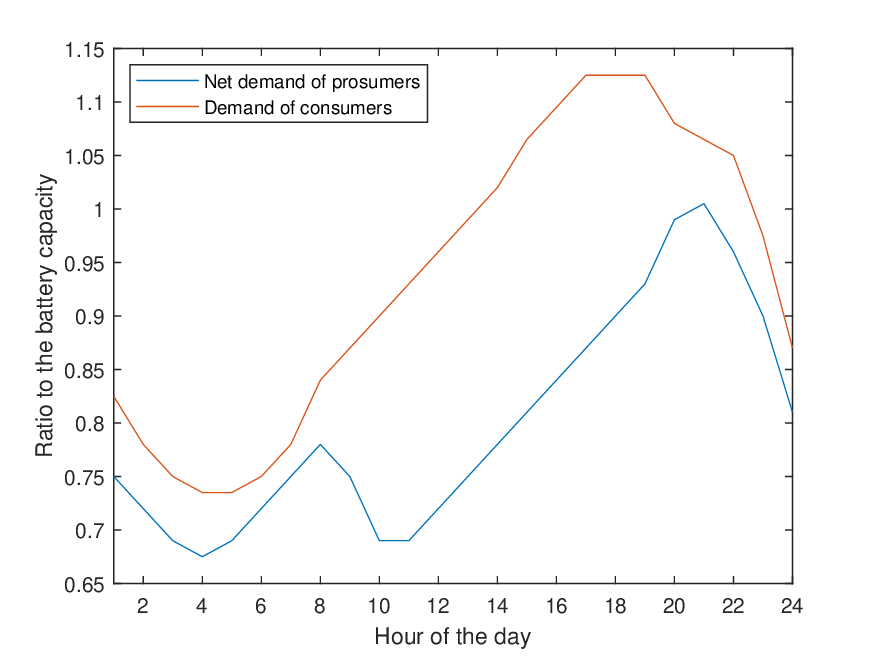}
    \caption{Daily load shapes for agents} \label{fig:shape}
\end{figure}
In our model, prosumers and consumers at the same bus are classified under the same location type, with all prosumers sharing a common net load shape and all consumers sharing a common gross load shape. To introduce variability across the 14 buses in our test network, we scale the CAISO load data by assigning each bus a unique load shape. Specifically, the base load profile is multiplied by a scaling factor, uniformly drawn from (0.9, 1.1), ensuring differentiation in demand patterns across buses.

To further introduce variability at the agent level, individual demand values in our simulation are sampled from a triangular distribution, with a lower bound of 0.8 times, an upper bound of 1.2 times the bus-specific average demand, and a mode equal to the average demand. This ensures that while all agents at a given bus follow the same general load pattern, their individual demand levels still vary, better capturing the diversity in consumption behaviors.

Each trading period in our simulation is set to one hour. To model demand surges, we introduce significant increases in energy demand between 6 PM and 9 PM on random days. This period aligns with the hours when solar output diminishes, providing an opportunity to evaluate how energy storage can be leveraged to mitigate early evening demand spikes within a completely decentralized decision-making framework. For supply shocks, we simulate increased generation, primarily driven by surges in distributed wind output, between 1 AM and 4 AM. These shocks occur on random days and are independent of demand shocks. The arrival of both demand and supply shocks is assumed to follow independent Poisson distributions, each with an arrival rate of 0.1 events per day. During a demand shock, the surge is represented as a percentage increase relative to typical demand, modeled using a triangular distribution with bounds [30\%, 50\%] and a mode of 40\%. Similarly, supply shocks involve increases in wind generation, modeled with a triangular distribution [20\%, 30\%] and a mode of 25\%. Agents are notified one hour in advance if the system operator anticipates a shock in the upcoming period.

The simulation includes 3,000 agents at each node, with each agent having a probability of 0.0001 of regenerating in each hour. To represent battery levels, the state space is discretized into 100 evenly spaced points between 0 and 100\%. 

\subsection{Result Analysis}
We run simulations using Algorithm \ref{Algo_meanfield} to model a 100-day period, repeated 10 times with different random seeds. Additionally,  we introduce two comparative scenarios: one where each agent maintains a single set of mean-field beliefs and does not adjust strategies in response to demand or supply shocks, and another without mean-field learning, where agents lack battery storage and bid solely based on their net load. This latter scenario represents a `grid-tied' setup in which solar or small wind generation is directly connected to the grid; any excess generation is immediately fed back into the grid without storage. 
Figure \ref{fig:diff} shows the average of relative difference between the belief of LMPs from an agent on Bus 3 and actual LMPs for Bus 3 across 10 runs. We select the LMPs from three typical hours -- 4 AM, 9 AM, and 9 PM -- when no supply or demand shocks occurred, as representative examples. 
The shaded region represents one standard deviation. It can be seen that the relative difference converges to almost zero quickly after about 10 days for all three hours, which indicates that agents' beliefs and their policies converge to a (mean-field) steady state quickly under our framework.
\begin{figure*}[!h]
   %\FIGURE
   \centering
    \includegraphics[width=\textwidth]{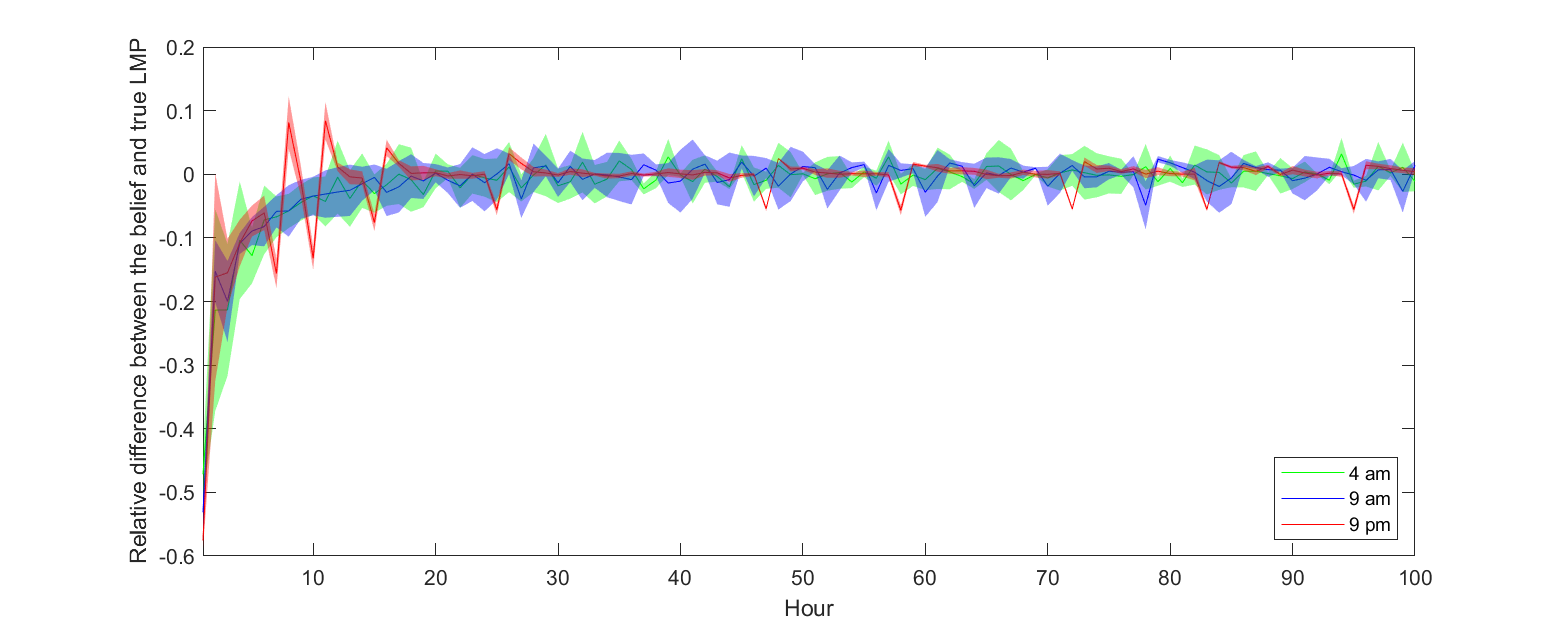}
    \caption{Relative difference between the belief and true LMP: mean-field learning without prior shock information}     \label{fig:diff}
\end{figure*}

\begin{figure*}[!h]
   %\FIGURE
   \centering
    \includegraphics[width=\textwidth]{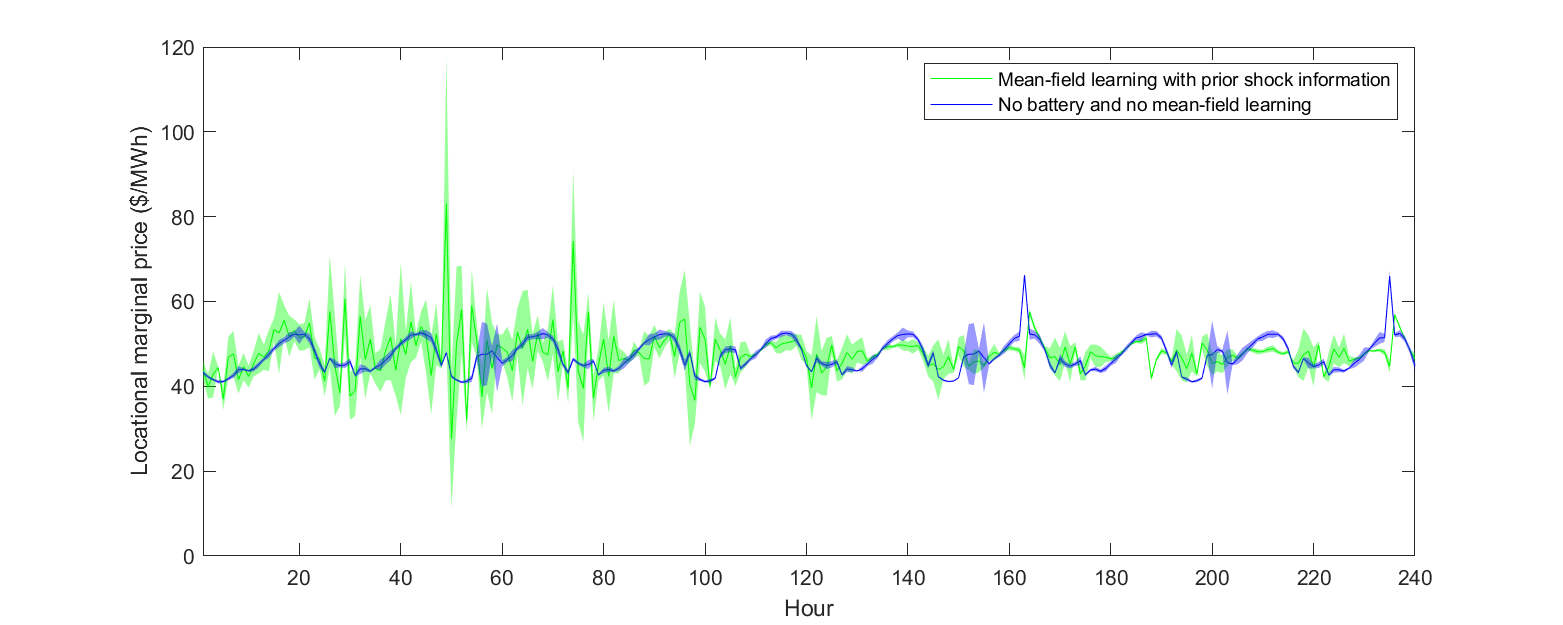}
   \caption{Hourly marginal prices of Bus 3 over the first 10 days: mean-field learning with prior shock information vs. no battery and no mean-field learning}
    \label{fig:f10days}
    {}
\end{figure*}
Figure \ref{fig:f10days} shows the realized LMPs for Bus 3, averaged over 10 runs, in chronological order for the first 10 days. The results indicate that LMPs stabilize quickly, reaching a steady state within less than 10 days, similar to the convergence pattern of LMP belief errors as in Figure \ref{fig:diff}. 
Compared to the LMPs in the scenario without energy storage (and therefore no agent learning), the LMPs with learning are lower during peak hours and higher during off-peak hours, resulting in reduced daily fluctuations.

\begin{figure*}[!h]
    %\FIGURE
    \centering
    \includegraphics[width=\textwidth]{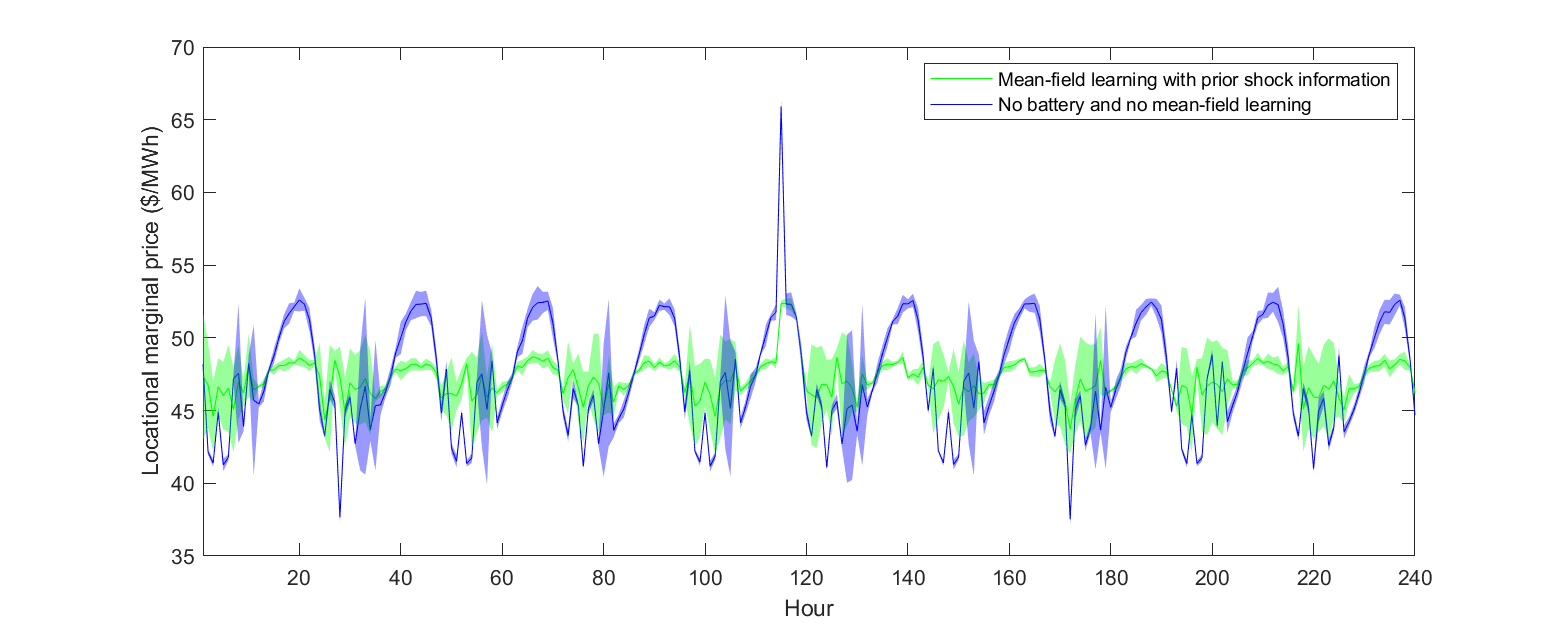}
   \caption{Hourly marginal prices of Bus 3 over the last 10 days: mean-field learning with prior shock information vs. no battery and no mean-field learning}
    \label{fig:l10days}
    {}
\end{figure*}
When there are supply or demand shocks, based on the results in Figure \ref{fig:l10days},  where the LMP for the last 10 days is presented in chronological order, the LMPs without mean-field learning can fluctuate dramatically. In comparison, with mean-field learning, the LMPs during demand shocks show only a slight increase from regular levels, while during supply shocks, they remain nearly unchanged.
%The LMPs are reduced from about \$63/MWh to about \$52/MWh during demand shocks. 
As outlined in Section \ref{sec:algo}, agents with prior shock arrival information adjust their LMP beliefs for demand or supply shocks upon receiving the corresponding signals. Since the LMPs anticipated by agents during demand shock scenarios are significantly higher than those under regular conditions, prosumers discharge more energy from their batteries according to their optimal strategies. Similarly, during supply shocks, when anticipated LMPs are lower, prosumers may choose to charge more energy into storage to take advantage of the lower prices.
\begin{figure}[!htb]
%\FIGURE
\centering
    \includegraphics[scale=0.8]{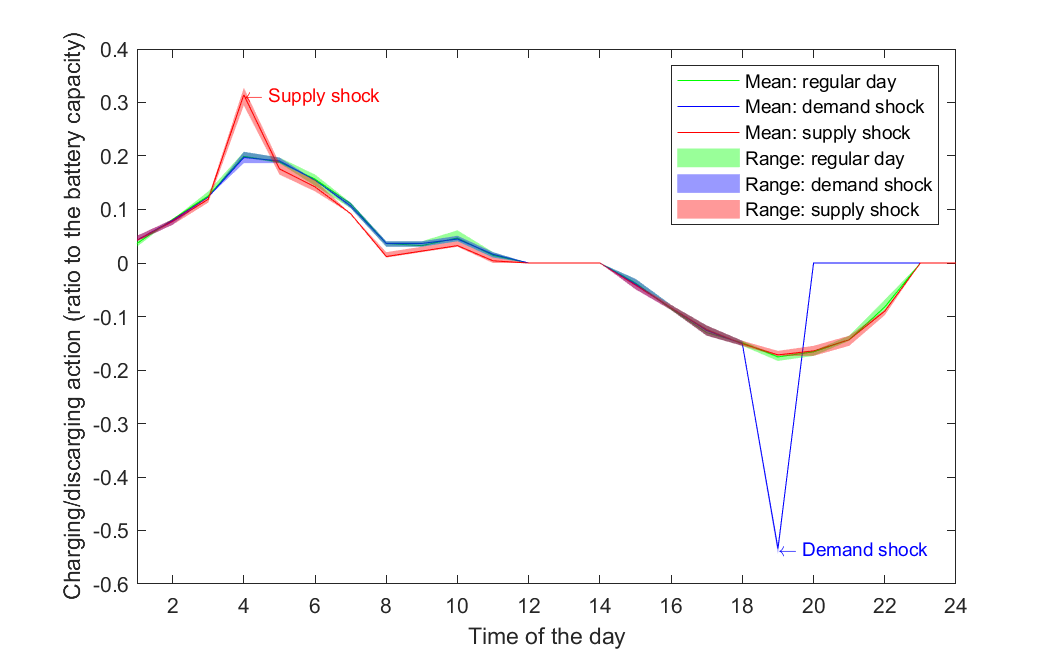}
   \caption{Charging/discharging actions over one day}
    \label{fig:action}    
\end{figure}

To evaluate how strategy adaptation during supply or demand shocks helps mitigate these events, we compare the performance of mean-field learning frameworks with and without prior knowledge of shock arrivals. This comparison is presented in Figures \ref{fig:f10days_shock} and \ref{fig:l10days_shock}, which display the average LMPs for Bus 3 over ten independent runs during the first and last ten days, respectively. While both frameworks perform similarly during regular hours, substantial differences arise during demand and supply shocks. Notably, the framework not using prior shock information and without pre-shock strategy adjustments struggles to manage significant price fluctuations during these critical periods, underscoring the importance of incorporating built-in mechanisms to address diverse emergency scenarios within the algorithm.

\begin{figure*}[!h]
%\FIGURE
\centering
    {\includegraphics[width=\textwidth]{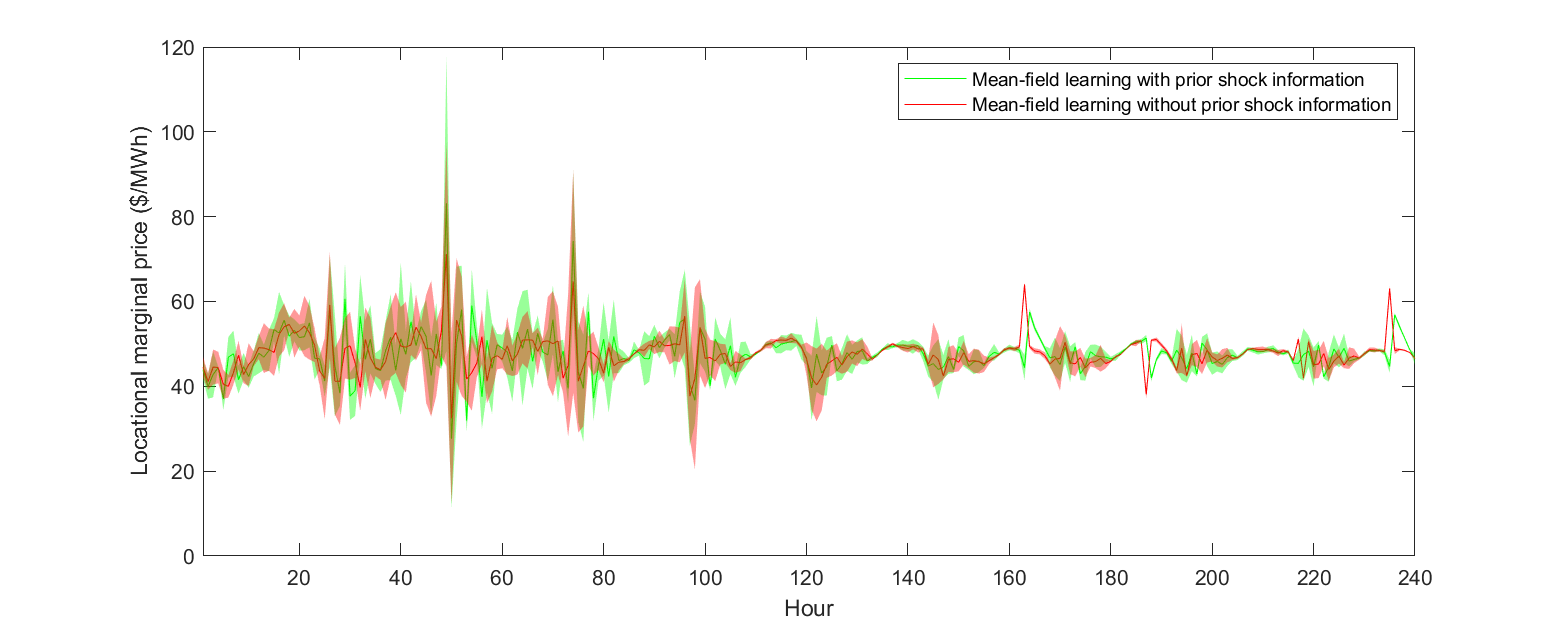}}
   \caption{Hourly marginal prices of Bus 3 over the first 10 days: mean-field learning with and without prior shock arrival information} 
    \label{fig:f10days_shock}    
    {}
\end{figure*}

\begin{figure*}[!h]
   %\FIGURE
   \centering
   \includegraphics[width=\textwidth]{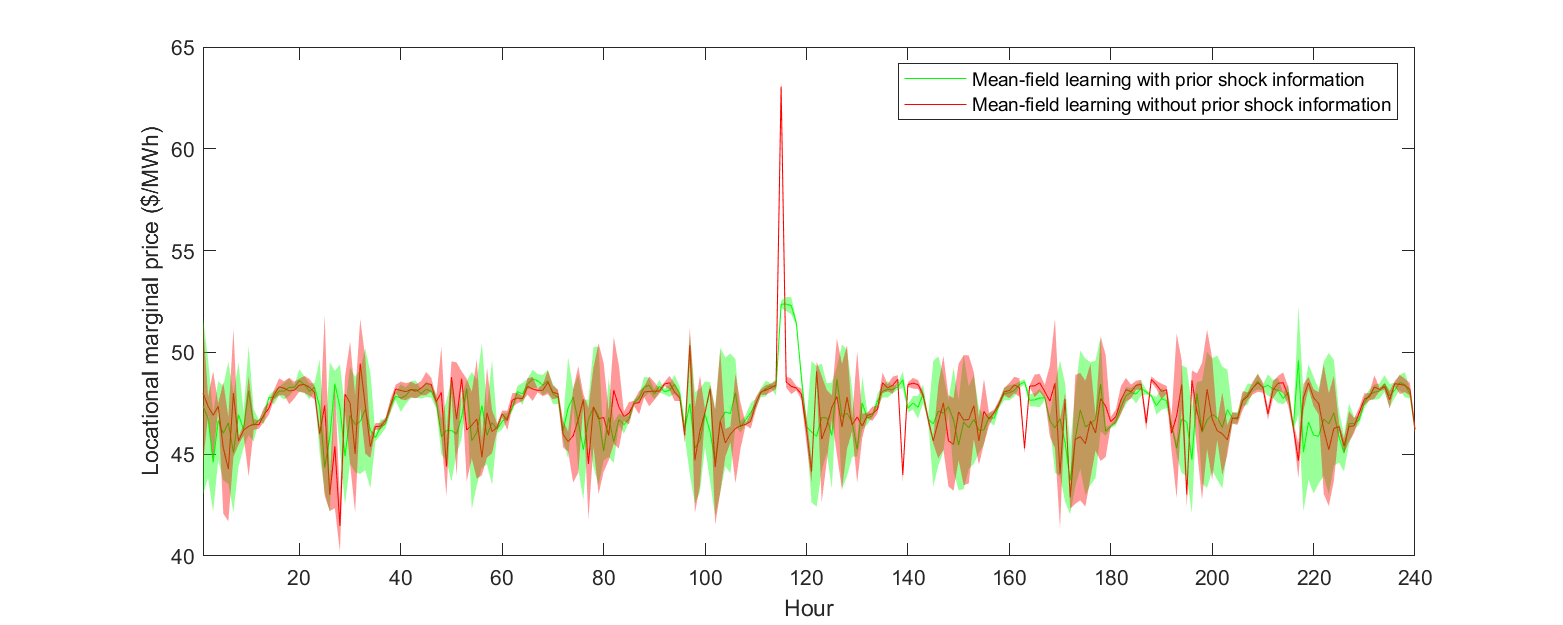}
 \caption{Hourly marginal prices of Bus 3 over the last 10 days: mean-field learning with and without prior shock arrival information}
    \label{fig:l10days_shock} 
\end{figure*}

To further compare the volatility across different cases, we adopt the volatility measure presented in \cite{roozbehani2012volatility}, which is the log-scaled incremental mean volatility (IMV). The IMV of a sequence $\{p_t\}_{t=1}^{\infty}$ is defined as
\begin{align}
\text{IMV} = \lim_{T\to\infty}\frac{1}{T}\sum_{t=1}^T |p_{t+1} - p_{t}|.
\end{align}
We approximate the IMV of a sequence of LMPs in our simulations using the prices from the last ten days, once the LMPs have reached a steady state. The average IMV over these ten days is computed as:  
$\overline{IMV} = \frac{1}{10} \sum_{i=1}^{10} IMV^i,$
where \( IMV^i \) represents the IMV of the \( i \)-th run. Table~\ref{IMV} presents the average IMVs at Bus 3, along with its standard deviation, over ten runs across three different learning approaches.
\begin{table*}[h!]
	\caption{%
		Averaged IMV of the LMPs at Bus 3 over 10 runs under three different scenarios}
	\vspace{0.2cm}
	\centering
	\begin{tabular}{l c c}
		\hline
		\centering \textbf{Scenario} & averaged IMV & Standard deviation\\
		\hline\hline
		\centering Mean-field
		learning with prior shock information & 0.348 &0.0038\\
		Mean-field
		\centering learning without prior shock information & 0.370&0.0035\\
		\centering No mean-field learning & 0.484& 0.0018\\
		
		\hline
	\end{tabular}
	\label{IMV}
\end{table*}
The results indicate that the scenario without mean-field learning exhibits greater volatility compared to the other two scenarios. Unsurprisingly, the mean-field learning framework with prior shock information achieves the lowest volatility, owing to its capacity to mitigate price fluctuations during shock hours effectively.

Finally, we compare the daily energy costs of all agents over the last ten days across 10 independent simulation runs, focusing on the scenario with mean-field learning and prior shock information versus the scenario without mean-field learning, as shown in Figure \ref{fig:energy_cost}. The results show a clear reduction in energy costs with mean-field learning.
\begin{figure*}[!h]
   %\FIGURE
   \centering
    \includegraphics[width=\textwidth]{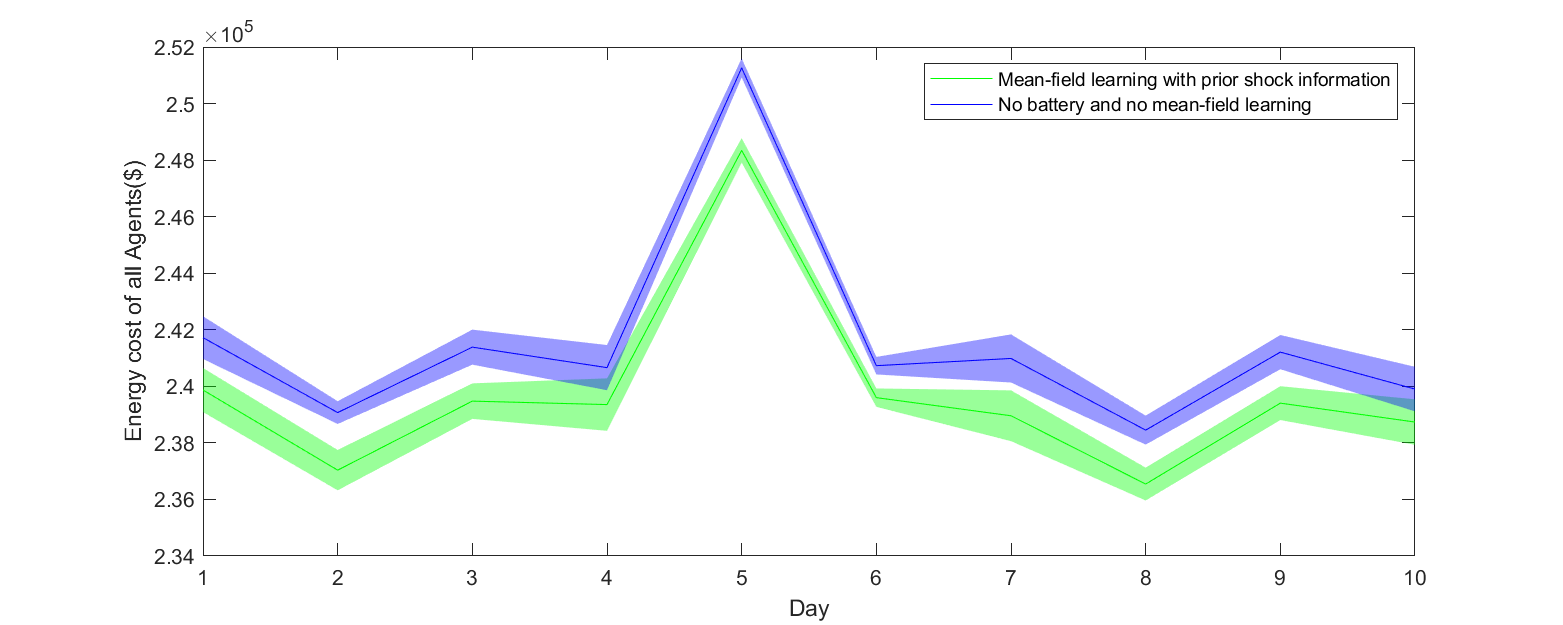}
    \caption{The total cost of all agents over the last 10 days: mean-field learning with prior shock information vs. no battery and no mean-field learning}
    \label{fig:energy_cost}
\end{figure*}

%Based on the numerical results, our framework demonstrates the convergence to a steady state and its capability to mitigate peak-hour LMPs and shock-hour price fluctuation with demand response. Additionally, it effectively reduces agents' energy cost through battery control.

\section{Conclusion and Future Research}
\label{sec:concl}

In this paper, we propose a mean-field game-based model and an algorithmic framework to enhance the participation of DER owners in wholesale energy markets. Our approach enables prosumers to make autonomous decisions based on real-time electricity prices while maintaining control over their assets. The mean-field approach is appropriate since all market information is reflected in the LMPs, which, along with signals from the system operator, are the primary data available to consumers and prosumers. We also proved the existence of a mean-field equilibrium for an infinite number of agents and the existence of an $\epsilon$-Markov-Nash equilibrium for a finite but sufficiently large number of agents within this framework. Our numerical results indicate that, even with high renewable penetration or extreme weather conditions, the decentralized learning approach can help prevent extreme LMP fluctuations, contributing to a more stable energy market.

An immediate extension of this work is to investigate whether a system with a finite number of agents can converge to the mean-field equilibrium (MFE) as the number of agents approaches infinity. Additionally, developing a provably convergent algorithm to reach the MFE remains an important area for further research.
Incorporating uncertainty in renewable generation and demand forecasts into prosumers' decision-making framework, and applying a reinforcement learning algorithm, could further enhance the robustness of the model under real-world conditions. Furthermore, if aggregators adopt a more active role, a promising direction is to apply mean-field control within each aggregator while modeling interactions among multiple aggregators as a mean-field game. Preliminary numerical results are provided in our related work \citep{HeLiu24}. We are currently working on establishing the theoretical foundations of this approach and will report our findings in a follow-up paper.

\appendix
\section{Proofs}\label{app:Proof}
\subsection{Proof of Proposition \ref{prop:LMP_LipCont} -- Lipschitz continuity of LMPs}
\label{Proof_LipCont}
To prove the Lipschitz continuity of the LMPs with respect to the aggregate demand, we will need to resort to linear complementarity problems (LCPs) and a known result regarding the Lipschitz continuity of LCP solutions. An LCP with a given vector \(u \in \mathbb{R}^n\) and a matrix \(M \in \mathbb{R}^{n \times n}\), denoted by \(\mathrm{LCP}(u, M)\), seeks to find an \(x \in \mathbb{R}^n\) such that \(0 \leq x \perp u + Mx \geq 0\), where the symbol \(\perp\) denotes orthogonality; that is, $x^T(u + Mx) = 0$.

\begin{theorem} \label{thm:Lip_LCP}
	(Theorem 3.2 in \cite{Mangasarian_LipCont} -- Lipschitz continuity of uniquely solvable LCPs).  
	Let $u^1$ and $u^2$ be points in $\Re^n$ such that the $\mathrm{LCP}(u(\tau),M)$ with $u(\tau):= (1 - \tau)u^1 + \tau u^2$ has a unique solution for each $\tau \in [0,1]$. Then the unique solutions $x^1$ of the LCP $(u^1, M)$ and $x^2$ of $(u^2, M)$ satisfy $||x^1 - x^2||_{\infty} \leq \sigma_{\beta}(M)||u^1 - u^2||_{\beta}$, where $\sigma_{\beta}(M)$ is some constant derived from the matrix $M$.    
\end{theorem}

\noindent\textbf{Proof of Proposition \ref{prop:LMP_LipCont}.} 
Since the LMPs are determined by the dual variables of supply and demand balancing constraint and the transmission line constraints, as given in \eqref{LMP}, under LICQ, the 
dual variables are unique, and hence, $P^n(\mathbf{B}_t)$ is single-valued with a given $\mathbf{B}_t \in \mathcal{F}_B.$

To utilize Theorem \ref{thm:Lip_LCP} to prove Lipschitz continuity of the LMPs with respect to energy demand, we write down the first-order optimality conditions (aka the KKT conditions) of the ED problem \eqref{obj1} -- \eqref{const3} at a given time $t$, with the quadratic cost function defined in \ref{prop:LMP_LipCont}:
\begin{align*}
0 \leq g_{t}^n & \perp  \alpha^n g_{t}^n + \beta^n - \lambda + \sum_{l=1}^L PTDF_{l}^n(\bar{\mu}_l - \underline{\mu}_l) + \bar{\eta}^n \geq 0 \\
0  \leq  \lambda \ & \perp  \sum_{n=1}^N g_{t}^n - \mathbf{1}^T  \mathbf{B}_t \geq 0 \\
0 \leq \bar{\mu}_l & \perp \widehat{F}_l - \sum_{n=1}^N PTDF_{l,n} (g_{t}^n - B_t^n) \geq 0 \\ 
0 \leq \underline{\mu}_l & \perp \widehat{F}_l + \sum_{n=1}^N PTDF_{l,n} (g_{t}^n - B_t^n) \geq 0\\
0 \leq \bar{\eta}_n & \perp \widehat{G}_n - g_t^n \geq 0.  
\end{align*}
Since the objective function in \eqref{obj1} is assumed to be convex quadratic, and the constraints are all linear (and hence the linear constraint qualification holds everywhere), the KKT condition is a necessary and sufficient optimality condition. 
%Let $\mathbf{g}_t = (g_t^1, \ldots, g_t^n)^T$, $\bm{\bar{\mu}} = (\bar{\mu}_1, \ldots, \bar{\mu}_L)^T$, $\bm{\underline{\mu}} = (\underline{\mu}_1, \ldots, \underline{\mu}_L)^T$, $\bm{\bar{\eta}} = (\bar{\eta}^1, \ldots, \bar{\eta}^N)^T$, $\bm{\alpha} = (\alpha^1, \ldots, \alpha^N)$, $\bm{\beta} = (\beta^1, \ldots, \beta^N)$, $\widehat{F} = (\widehat{F}_1,\ldots, \widehat{F}_L)^T$, $\widehat{G} = (\widehat{G}_1, \ldots, \widehat{G}_N)^T$. 
Let $\mathbf{g}_t , \bm{\bar{\mu}}, \bm{\underline{\mu}}, \bm{\bar{\eta}}, \bm{\alpha}, \bm{\beta}, \bm{\widehat{F}},$ and $\bm{\widehat{G}}$ represent vectors containing collections of their corresponding elements. Furthermore, let $PTDF \in \Re^{L \times N}$ be the matrix whose $l$-th row and $n$-th column element is $PTDF_{l}^n$.
Furthermore, let $PTDF \in \Re^{L \times N}$ be the matrix whose $l$-th row and $n$-th column is $PTDF_{l}^n$, and $\Lambda = \mathrm{Diag}(\bm{\alpha}) \in \Re^{N \times N}$ be a diagonal matrix with diagonal entries being the elements of the vector $\bm{\alpha}$. We can write the KKT conditions into the following LCP form:
\begin{align*}
&0\leq \begin{pmatrix}
\mathbf{g}_t \\
\lambda\\
\bm{\bar{\mu}}\\
\bm{\underline{\mu}}\\
\bm{\bar{\eta}}
\end{pmatrix} \perp 
\begin{pmatrix}
\bm{\beta}\\
- \mathbf{B}^t \\
\bm{\widehat{F}} + PTDF \times \mathbf{B}^t\\
\bm{\widehat{F}} - PTDF \times \mathbf{B}^t\\
\bm{\widehat{G}}\\
\end{pmatrix} + \begin{bmatrix}
\Lambda & - \mathbf{1} & PTDF & - PTDF & I \\
\mathbf{1}^T & 0 & 0 & 0 & 0 \\
-PTDF & 0 & 0 & 0 & 0 \\
PTDF & 0 & 0 & 0 & 0 \\
- I & 0 & 0 & 0 & 0 \\
\end{bmatrix} 
\begin{pmatrix}
\mathbf{g}_t \\
\lambda\\
\bm{\bar{\mu}}\\
\bm{\underline{\mu}}\\
\bm{\bar{\eta}}
\end{pmatrix} \geq 0, 
\end{align*}
where $\mathbf{1}$ denotes a vector of all 1's, $I$ denotes the identity matrix, and $0$ represents either a vector or a matrix, all of the appropriate dimensions. 
Let $\mathbf{x}$ denote the collection of all variables in the above LCP, 
$\mathbf{u(B^t)}$ represent the constant vector, and $M$ be the big matrix.  
Then, the LCP above can be written in the following condensed form:
\begin{equation}\label{eq:LCP_dense}
0 \leq \mathbf{x} \perp \mathbf{u(B^t)} + M\mathbf{x} \geq 0. 
\end{equation}
Under the assumptions of a strongly convex objective function and Assumption \ref{assump:LICQ}, for a given $\mathbf{B}^t$, the optimal primal and dual solutions are unique, and hence, the LCP \eqref{eq:LCP_dense} also has a unique solution. Consequently, Theorem \ref{thm:Lip_LCP} applies here, and since $\mathbf{u(B^t)}$ is a linear function with respect to $B^t$, it is straightforward to derive the LMPs, $P^n(\mathbf{B^t})$ as defined in \eqref{LMP}, are Lipschitz continuous with respect to $\mathbf{B^t}$ for $n = 1, \ldots, N$. 
\hfill$\Box$

\subsection{Proof of single-valuedness of a prosumer's optimal policy} \label{subsec:PolicyUnique}
To facilitate the derivation of theoretical results that follow, we need to endow $\mathcal{P}(\Xi)$ with the weak topology through the concept of weak convergence as follows.

\begin{definition}
	\label{def:weak}
	(Weak convergence \citep{guide2006infinite}) 
	We say that a sequence of measures $\{p_n\} \in \mathcal{P}(\Xi)$ converges weakly to $p \in \mathcal{P}(\Xi)$ if, for all bounded and continuous functions $f:\Xi \to \mathbb{R}$, we have
	\begin{align*}
	\lim _{n \to \infty}\int_{\Xi} f(x) \, p_n(dx) = \int_{\Xi} f(x) \, p(dx).
	\end{align*}
\end{definition}

To prove Proposition \ref{prop:unique}, we need to first show that the expectation of the LMPs is continuous with respect to the population \(p^{\infty}_{t}\) at any location and at any time period \(t\).

\begin{lemma}
	Under Assumption \ref{assump:LICQ} and the condition that at time $t$, the random noise of individual agent's demand \(\zeta^{\theta}_{i,t}\), as defined in \eqref{eq:q_decompose}, is i.i.d. 
    %with the mean denoted by $\bar{\zeta}^{\theta}_t$, 
    the expected value of the LMP at each node \(n = 1, \ldots, N\) at time \(t\), as defined in Eq. (\ref{LMP}), is a continuous function of the population profile \(p^{\infty}_{t}\) with respect to weak convergence, as defined in Definition \ref{def:weak}.
	\label{lemma:ELMP}
\end{lemma}
\begin{proof}{Proof} As in Eq. (\ref{LMP}), the LMPs at each node $n$ are a function of the aggregated demand bids at all locations. Based on the definition of the bids in \eqref{eq:bids}, when $I^{\theta} \to \infty$ for all \(\theta \in \Theta\), since the total capacity of all type $\theta$ agents is assumed to be capped at $\overline{C}^{\theta}$, each individual agent's bid becomes infinitesimal, and the aggregate bids remain finite. We first characterize such aggregate bids using the Strong Law of Large Numbers (SLLN).  

To sum over all the bids, for ease of notation, we use a function,  $v^{\theta}_{i,t}(e_{i,t}, a_{i,t})$, to denote the second part of an agent's bid in \eqref{eq:bids}: 
\begin{align}\label{eq:v_function}
& v^{\theta}_{i,t}(e_{i,t}, a_{i,t}):= \left\{
\begin{aligned}
& \eta(a_{i,t}) \cdot \max \big\{- e_{i,t},\ a_{i,t} \big\}, \    a_{i,t} < 0,\\
&  \frac{\min \big\{1-e_{i,t},\ a_{i,t}\big\}}{\eta(a_{i,t})}, \ a_{i,t} \geq 0. 
\end{aligned}
\right.
\end{align}
Since both the state and action are random variables (due to the exogenous uncertainties in each agent's reward functions), whose joint distribution is exactly the population profile $p_t^{\infty,\theta}$ when $I^{\theta} \to \infty$,  $v^{\theta}_{i,t}(e_{i,t}, a_{i,t})$ is also a random variable. 
Since within the same type, all agents use the same optimal policy and are subject to the same weather conditions, we can assume that the series $\{v^{\theta}_{i,t}\}_{i=1}^{\infty}$ is i.i.d. By multiplying $\bar{e} = \overline{C}^{\theta}/{I^{\theta}}$ (to obtain the actual energy bids considering battery charging/discharging, as defined in the bid formulation \eqref{eq:bids} and by applying the SLLN, we have that 
\begin{equation}\label{eq:sum_v}
\begin{array}{l}
\displaystyle
\lim_{I^{\theta}\to \infty}\sum_{i=1}^{I^{\theta}} \left[v^{\theta}_{i,t}(e_{i,t}, a_{i,t}) \times   \bar{e}\right] = \lim_{I^{\theta}\to \infty}\Bigg[\frac{\displaystyle\sum_{i=1}^{I^{\theta}} v^{\theta}_{i,t}(e_{i,t},a_{i,t})}{I^{\theta}} \Bigg]\overline{C}^{\theta}  =  \ \overline{C}^{\theta}\int_{\mathcal{E}\times\mathcal{A}} v^{\theta}_{i,t}(e_{i,t},a_{i,t})dp_t^{\infty, \theta},  
\end{array}
\end{equation}
where the integration in the last equation represents the expected value of $v^{\theta}_{i,t}(e_{i,t},a_{i,t})$.\footnote{Note that in \eqref{eq:sum_q}, when \(I^{\theta} \to \infty\), it does not imply that the agents' actions (and states) correspond to a finite-agent game with \(I^{\theta}\) agents. Instead, the agents' actions are derived from the optimal policy in the setting where the number of agents is already infinite. The limit in \eqref{eq:sum_q} simply represents the partial sum of an infinite series.}

For the other part of an agent's bid,  $q^{\theta}_{i,t}\bar{e}$, we have that 
\begin{align}
& \lim_{I^{\theta}\to \infty}\sum_{i=1}^{I^{\theta}}  q^{\theta}_{i,t}\bar{e} = \lim_{I^{\theta}\to \infty}\sum_{i=1}^{I^{\theta}}\left(\omega_t^{\theta} + \zeta_{i,t}^{\theta}\right)\frac{\overline{C}^{\theta}}{I^{\theta}} = \left(\omega^{\theta}_t + \lim_{I^{\theta}\to \infty}\frac{\sum_{i=1}^{I^{\theta}}\zeta_{i,t}^{\theta}}{I^{\theta}}\right)\overline{C}^{\theta} = 
\left(\omega^{\theta}_t + \bar{\zeta}_{t}^{\theta} \right)\overline{C}^{\theta}, \label{eq:sum_q}
\end{align}
where the second equality holds because the random variable \(\omega_t^{\theta}\) represents weather-related uncertainties and does not depend on the agents (hence, no agent subindex \(i\)). The last equality directly follows from the SLLN.

By \eqref{eq:sum_v} and \eqref{eq:sum_q}, with a given population profile $p_t^{\infty,\theta}$, we can write out the aggregate bids of type $\theta$ as follows: 
\begin{align}\label{eq:agg_bids}
\begin{split}
B^{\infty,\theta}_{t} := & \sum_{i=1}^{\infty}b_{i,t}^{\theta}(e_{i,t},a_{i,t},q^{\theta}_{i,t})  \\
= &  \lim_{I^{\theta}\to \infty}\sum_{i=1}^{I^{\theta}} \left[v^{\theta}_{i,t}(e_{i,t}, a_{i,t}) +\omega_t^{\theta} + \zeta_{i,t}^{\theta}\right] \times   \bar{e} \\
= & \ \overline{C}^{\theta}\left(\int_{\mathcal{E}\times\mathcal{A}} v^{\theta}_{i,t}(e_{i,t},a_{i,t})dp_t^{\infty, \theta} +\omega^{\theta}_t + \bar{\zeta}_{t}^{\theta} \right).
\end{split}
\end{align}
It can be seen that the aggregate bids for each type remain random variables due to the presence of the weather-related random variable \(\omega^{\theta}\). Let \(\rho^{\omega}\) denote the joint probability distribution of \(\omega^{\theta}\) for all \(\theta \in \Theta\), and assume the joint distribution has a compact support \(\Omega\). Using the formulation in \eqref{LMP}, the expected value of the LMP at time \(t\) at node \(n = 1, \ldots, N\) can be written as:
\begin{align}\label{eq:ELMP}
\begin{split}
&\mathbb{E}[LMP^n_t] = \int_{\Omega} P^n\bigg(B^{\infty,1}_{t},\cdots, B^{\infty,N}_{t}\bigg) \, \rho^{\omega}(d\omega).\\
\end{split}
\end{align}

Let \(\{p^{\infty,\theta}_{t,k}\}_{k=1}^{\infty}\) be a sequence of population measures that weakly converge to \(\{p^{\infty,\theta}_{t}\}\). By its definition in \eqref{eq:v_function}, the function \(v_{i,t}^{\theta}(e_{i,t}, a_{i,t})\) is bounded and continuous. Hence, by Definition \ref{def:weak}, we have
\begin{align*}
\lim_{k\to \infty} \int_{\mathcal{E}\times\mathcal{A}} v^{\theta}_{i,t}(e_{i,t}, a_{i,t}) \, dp^{\infty, \theta}_{t, k}  = \ \int_{\mathcal{E}\times\mathcal{A}} v^{\theta}_{i,t}(e_{i,t}, a_{i,t}) \, dp_t^{\infty, \theta}.
\end{align*}
As a result, by \eqref{eq:agg_bids} and \eqref{eq:ELMP}, \(\mathbb{E}[LMP^n_t]\) is continuous with respect to the population profile \(\{p^{\infty,\theta}_{t}\}\) since the LMP, \(P^n\), is Lipschitz continuous with respect to the aggregated bids under Assumption \ref{assump:LICQ}. %\Halmos
\end{proof}

\noindent \textbf{Proof of Proposition \ref{prop:unique}.}
	The non-emptiness of the mapping in \eqref{eq:OptPolicy} follows from the existence of a stationary optimal policy, a well-established result in dynamic programming, as mentioned earlier. Therefore, we omit the proof and proceed to show that the objective function in the Bellman equation \eqref{eq:Bellman} is strictly concave with respect to $a$. Together with the non-emptiness of the mapping, this implies the single-valuedness of the optimal policy.

	To do so, we want to simplify the bid function \eqref{eq:bids} by removing the outer `max' or `min' operator first. By writing out the bid function explicitly and restricting the action $a$ based on the current energy storage level $[-e, 1-e]$, we can equivalently re-write the Bellman equation as follows:
	\begin{align}
	& V^{\pi^{\theta^*}}(s, p^{\infty})  \nonumber \\
    = &\ \max_{a\in \mathcal{A}} \left\{\overbar{R}^{\theta}(s, a|p^{\infty}) + \beta  V^{\pi^{\theta^*}}\left[Tr(s,a),\ p^{\infty}\right]\right\} \nonumber \\
	= &   \max_{a\in [-e,\ 1-e]} \bigg\{\mathbb{E}_{q^{\theta}} \left[q^{\theta}\overline{e}^{\theta}\right] -\ {P}^{n(\theta)}(p^{\infty}) \cdot 
	\eta(a)\cdot\overline{e}^{\theta} \cdot \min(a, 0)  -\ {P}^{n(\theta)}(p^{\infty})  \cdot 
	\left(\overline{e}^{\theta}/\eta(a)\right) \cdot \max(a, 0) \nonumber \\
	& +\ \beta V^{\pi^{\theta^*}}\left[Tr(s,a), p^{\infty}\right]\bigg\}.  \label{eq:Bellman_Reform}
	\end{align}
	We want to show that $\overbar{R}^{\theta}(s, a|p^{\infty})$ is strictly concave with respect to $a$. Since the first term in \eqref{eq:Bellman_Reform}, $\mathbb{E}_{q^{\theta}} [q^{\theta}\overline{e}^{\theta}]$ is a constant, we only need to focus on the remaining terms. Additionally, for a given population profile $p^{\infty}$, the LMP ${P}^{n(\theta)}(p^{\infty})$ can also be treated as a constant for a given agent, which we simply denote it as ${P}$. Consider the following step-wise function:
	\begin{align}\label{eq:u_function}
	u^{\theta}(a)&:= \begin{cases}
	-\eta(a) \cdot a \cdot \Bar{e}^{\theta} \cdot {P} &\text{ if } a \in[-1,0],\\[3pt]
	- \displaystyle\frac{a}{\eta(a)}\cdot   \Bar{e}^{\theta} \cdot {P} &\text{ if } a \in (0,1],
	\end{cases} \\[5pt]
	& \ =  \begin{cases}
	-(\alpha_0 + \alpha_d \cdot a)  \cdot a \cdot \Bar{e}^{\theta} \cdot {P}&\text{ if } a \in[-1,0],\\[3pt]
	-  \displaystyle \frac{a}{\alpha_0 - \alpha_c \cdot a} \cdot \Bar{e}^{\theta}\cdot  {P} &\text{ if } a \in (0,1],
	\end{cases}  \label{eq:reward_no_e}
	\end{align}
	where $\alpha_0 \in (0, 1)$, $\alpha_c$,  and $\alpha_d > 0$ are the parameters in defining battery charging/discharging efficiency in \eqref{eq:eff}, with $\alpha_0 - \alpha_c > 0$ and $\alpha_0 - \alpha_d > 0$. 
	It is straightforward to see that $u^{\theta}(a)$ is strictly concave on either [-1, 0] or on (0, 1]. To show that $u(a)$ is strictly concave over the entire region [-1, 1], we construct two auxiliary functions $\overline{u}^{\theta}(a)$ and $\widetilde{u}^{\theta}(a)$ as follows:  
	\begin{align}
	& \overline{u}^{\theta}(a) := \begin{cases}
	-(\alpha_0 + \alpha_d \cdot a) \cdot a  \cdot \Bar{e}^{\theta} \cdot \widetilde{P}, &\ \text{ if } a\in[-1,0],\\[3pt]
	[(\frac{1}{2}\alpha_0-\frac{1}{2\alpha_0}) \cdot a^2 - \alpha_0  a] \cdot \Bar{e}^{\theta} \cdot \widetilde{P}, &\ \text{ if } a \in (0,1],
	\end{cases} \\[5pt]
	& \mathrm{and}  \nonumber \\[5pt]
	&\widetilde{u}^{\theta}(a)  :=  \begin{cases}
	[(\frac{1}{2}\alpha_0-\frac{1}{2\alpha_0}) \cdot a^2 - \frac{1}{\alpha_0}  a] \cdot \Bar{e}^{\theta} \cdot \widetilde{P},   &\ \text{ if } a \in[-1,0],\\[3pt]
	\displaystyle -  \frac{a}{\alpha_0 - \alpha_c \cdot a} \cdot \Bar{e}^{\theta} \cdot \widetilde{P},  &\ \text{ if } a_h  \in (0,1].
	\end{cases} 
	\end{align} 
	By taking the derivatives of the two functions, we get that 
	
	\begin{align}
	& \frac{d\overline{u}^{\theta}(a)}{da}  = \begin{cases}
	-(\alpha_0 + 2\alpha_d a) \cdot \Bar{e}^{\theta} \cdot \widetilde{P},&\ \text{ if } a\in[-1,0],\\[3pt]
	[(\alpha_0-\frac{1}{\alpha_0}) \cdot a - \alpha_0  ] \cdot \Bar{e}^{\theta} \cdot \widetilde{P},   &\ \text{ if } a \in (0,1],
	\end{cases} \\[5pt]
	& \mathrm{and}  \nonumber \\[5pt]
	&\frac{d \widetilde{u}^{\theta}(a)}{d a} =  \begin{cases}
	[(\alpha_0-\frac{1}{\alpha_0}) \cdot a - \frac{1}{\alpha_0} ] \cdot \Bar{e}^{\theta} \cdot \widetilde{P}, &\ \text{ if } a \in[-1,0],\\[5pt]
	\displaystyle -  \frac{\alpha_0}{(\alpha_0 - \alpha_c \cdot a)^2} \cdot \Bar{e}^{\theta} \cdot \widetilde{P}, &\ \text{ if } a  \in (0,1].
	\end{cases} 
	\end{align} 
	Note that both functions are differentiable over the entire range of \([-1,1]\), as the left and right derivatives at \(a = 0\) are equal for both functions. For $\overline{u}^{\theta}(a)$, when $a\in [-1,0]$, clearly $d\overline{u}^{\theta}(a)/da$ is a strictly decreasing function since $\alpha_d$, $\Bar{e}^{\theta}$, and ${P}$ are all positive. 
	When $a\in (0,1]$, since $\alpha_0 \in (0, 1)$, then $d\overline{u}^{\theta}(a)/da$ is also a strictly decreasing function. Hence, $d\overline{u}^{\theta}(a)/da$ is strictly decreasing over $[-1, 1]$. By the well-known result for univariate functions (see Theorem 1.4 in \cite{ConvexFunctions}), $\overline{u}^{\theta}(a)$ is strictly concave on $[-1, 1]$. Similarly, we can show that $\widetilde{u}^{\theta}(a)$ is also strictly concave on $[-1, 1]$. 
	
	By the way of constructing $\overline{u}^{\theta}$ and $\widetilde{u}^{\theta}$, it is easy to see that 
	$u^{\theta} (a) = \min \{\overline{u}^{\theta}(a), \tilde{u}^{\theta}(a)\}$. Hence, $u^{\theta} (a)$ is strictly concave on $[-1, 1]$. Next, we show that the optimal value function $V^{\pi^{\theta^*}}(Tr(s,a),p^{\infty})$ is also strictly concave in $a$.
	
	Let $\overline{\mathcal{J}}(\mathcal{E} \times \mathcal{H} \times \mathcal{P}(\Xi)^{|\Theta|})$ denote the space of all bounded functions on $\mathcal{E} \times \mathcal{H} \times \mathcal{P}(\Xi)^{|\Theta|}$, where $\mathcal{E} = [0, 1]$ is the range of the energy storage state of charge, $\mathcal{H}$ is the discrete set of all times of day, and $\mathcal{P}(\Xi)^{|\Theta|}$ is the space of possible distributions of population profile $p^{\infty}$.  For a function $J^{\theta}(s, p^{\infty}) \in \overline{J}$ that is jointly continuous, define the Bellman operator $T: \overline{\mathcal{J}} \to \overline{\mathcal{J}}$  as follows: 
	\begin{align}
	TJ^{\theta}(s, p^{\infty}) = \ \max_{a\in\mathcal{A}}\overbar{R}^{\theta}(s, a|p^{\infty}) + \beta  J^{\theta}\left(Tr(s,a), p^{\infty}\right). \label{eq:BellmanOperator}
	\end{align}
	Although the state variable includes both the state of charge and the time of day, we can focus solely on the state of charge, as the time of day transition is discrete and deterministic, and it will not affect any of the discussion that follows. 
	To simplify the transition function of the state of charge \eqref{eq:StateTransition}, we can let the feasible action space depend on the current state of charge, that is $a \in [-e, 1-e]$, then the state transition function \eqref{eq:StateTransition} becomes $E(s, a)= e + a$. 
	Let $J^{\theta}$ be any continuous function on $\overline{J}$ and concave with respect to $s$, then  $J^{\theta}(E(s,a),p^{\infty})$ is also concave with respect $a$ since $E(s,a)$ is a linear function in $s$ and $a$. Now define the Bellman operator corresponding to the modified Bellman equation \eqref{eq:Bellman_Reform} as follows: 
	\begin{align}
	& TJ^{\theta}(s, p^{\infty})  = \ \max_{a\in [-e, \ 1-e]}\overbar{R}^{\theta}(s, a|p^{\infty}) + \beta  J^{\theta}\left(E(s,a), p^{\infty}\right). \label{eq:BellmanOperator_Reform}
	\end{align}
	By reformulating the reward function as in \eqref{eq:Bellman_Reform} and expressing its explicit form in \eqref{eq:reward_no_e}, the reward function $\overbar{R}^{\theta}(s, a|p^{\infty})$ does not explicitly depend on the state variable $s$. Since we have shown that it is concave in $a$, the term $\overbar{R}^{\theta}(s, a|p^{\infty}) + \beta  J^{\theta}\left(E(s,a), p^{\infty}\right)$ is jointly concave in $(s,a)$. Additionally, the feasible region \( a \in \mathcal{A}(e) \equiv [-e, 1-e] \), considered as a point-to-set mapping, is hull concave over the percentage interval \( \mathcal{E} = [0, 1] \), meaning that the convex hull of $\mathcal{A}(e)$ is a concave mapping over \( \mathcal{E} \).
	By a well-known result on the concavity of optimal value functions (see Proposition 3.2 in \cite{fiacco1986convexity}), $TJ^{\theta}$ is concave in $s$ for a fixed $p^{\infty}$. Consequently, the operator $T$ preserves concavity, and $T^kJ^{\theta}$ remains concave in $s$ for all $k= 1, 2, \ldots .$
	Furthermore, by the standard result from dynamic programming, the Bellman operator is a contraction mapping, ensuring that $T^k J^{\theta}$ converges uniformly to $V^{\pi^{\theta^*}}$ (see \cite{bertsekas1996stochastic}). Therefore, by a known result in convex analysis stating that the pointwise limit of a sequence of convex functions is also convex (Theorem 10.8 in \cite{Rockafellar}), \( V^{\pi^{\theta^*}}(s, p^{\infty}) \) is concave with respect to \( s \), implying that $V^{\pi^{\theta^*}}[Tr(s,a), p^{\infty}]$ is concave with respect to $a$. Together with the strict concavity of the function $u^{\theta}(a)$ in \eqref{eq:u_function} (and thus the strict concavity of $\overbar{R}^{\theta}(s, a|p^{\infty})$ in $a$), the `argmax' mapping in \eqref{eq:OptPolicy} must be a singleton.
	
	To show that the optimal policy mapping is continuous in \( (s, p^{\infty}) \), we again rely on the Bellman operator in \eqref{eq:BellmanOperator} with an arbitrary continuous function \( J^{\theta} \in \overline{J} \). The one-stage reward function \( \overbar{R}^{\theta}(s, a \mid p^{\infty}) \) is the product of the LMP and bid quantity. By Proposition \ref{prop:LMP_LipCont}, the LMP is Lipschitz continuous with respect to \( p^{\infty} \). Since the bid function is jointly continuous in \( (s,a) \) (as can be seen in \eqref{eq:bids}), the reward function is jointly continuous in \( [(s, p^{\infty}), a] \) in light of Lemma \ref{lemma:ELMP}. Furthermore, the transition function \( Tr(s, a) \), as defined in \eqref{eq:StateTransition}, is also jointly continuous in \( (s, a) \), making \( J^{\theta}(Tr(s,a), p^{\infty}) \) jointly continuous as well, given that \( J^{\theta} \) is a continuous function.
	
	With the feasible action space \( \mathcal{A} \) being compact, the Berge Maximum Theorem (Theorem 17.31 in \cite{guide2006infinite} or Lemma 6.11.8 in \cite{Puterman_MDP}) ensures that the optimal value function \( TJ^{\theta}(s, p^{\infty}) \) is continuous in \( (s, p^{\infty}) \). Since \( T^k J^{\theta} \) converges uniformly to \( V^{\pi^{\theta^*}} \), the uniform limit theorem (Theorem 21.6 in \cite{Topology}) guarantees that \( V^{\pi^{\theta^*}}(s, p^{\infty}) \) is jointly continuous. Finally, by the Berge Maximum Theorem again (or Lemma 6.11.9 in \cite{Puterman_MDP}), the unique `argmax' in \eqref{eq:OptPolicy} is continuous in \( (s, p^{\infty}) \). %\Halmos

\subsection{Proof of MFE existence}\label{subsec:Proof_MFEExist}
As stated in the main text, proving the existence of an MFE in our context requires the Schauder-Tychonoff Fixed Point Theorem, stated below

\begin{proposition}
	(Schauder-Tychonoff Fixed Point Theorem, Corollary 17.56, \cite{guide2006infinite}) Let $X$ be a nonempty, compact, convex subset of a locally convex Hausdorff space, and let $f: X \rightarrow X$ be a continuous function. Then the set of fixed points of $f$ is compact and nonempty.
\end{proposition}

\noindent \textbf{Proof of Proposition \ref{prop:MFE_exist}.} Given the uniform boundedness of the reward function (Remark 1) and the continuity result from Proposition \ref{prop:unique}, the existence proof follows directly from Theorem 3 in \cite{light2022mean}, which applies the Schauder-Tychonoff Fixed Point Theorem. %\Halmos

\bibliographystyle{informs2014}
\bibliography{MFG}

%%%%%%%%%%%%%%%%%
\end{document}